\documentclass[journal,10pt]{IEEEtran}
\usepackage{verbatim}
\usepackage{amsfonts,amssymb,amsmath}
\usepackage{hyperref}
\usepackage{multirow}
\usepackage{graphicx}
\usepackage{caption}
\usepackage{amsfonts,amssymb,amsmath}
\usepackage{hyperref}
\usepackage{embrac}
\newtheorem{theorem}{Theorem}
\newtheorem{proof}{Proof}
\usepackage{subeqnarray}
\usepackage{cases}
\usepackage{algorithm}
\usepackage{algpseudocode}
\usepackage{cuted}%%\stripsep-3pt
\usepackage{tabularx} 
\makeatletter
\newenvironment{breakablealgorithm}
{% \begin{breakablealgorithm}
		\begin{center}
			\refstepcounter{algorithm}% New algorithm
			\hrule height.8pt depth0pt \kern2pt% \@fs@pre for \@fs@ruled
			\renewcommand{\caption}[2][\relax]{% Make a new \caption
				{\raggedright\textbf{\ALG@name~\thealgorithm} ##2\par}%
				\ifx\relax##1\relax % #1 is \relax
				\addcontentsline{loa}{algorithm}{\protect\numberline{\thealgorithm}##2}%
				\else % #1 is not \relax
				\addcontentsline{loa}{algorithm}{\protect\numberline{\thealgorithm}##1}%
				\fi
				\kern2pt\hrule\kern2pt
			}
		}{% \end{breakablealgorithm}
		\kern2pt\hrule\relax% \@fs@post for \@fs@ruled
	\end{center}
}
\makeatother

\begin{document}

\title{Opportunistic Routing aided Cooperative Communication MRC Network with Energy-Harvesting Nodes}

\author{Lei Teng, Wannian An, Chen Dong, Xiaodong Xu, Boxiao Han

\thanks{This work was supported in part by State Key Laboratory of Networking and Switching Technology, Beijing University of Posts and Telecommunications, and in part by Beijing University of Posts and Telecommunications-China Mobile Research Institute Joint Innovation Center (Project Number: R207010101125D9).\emph{(Corresponding author: Chen Dong.)}}
\thanks{Lei Teng, Wannian An, Chen Dong and Xiaodong Xu are with the School of Information and Communication Engineering Beijing University of Posts and Telecommunications(e-mail: tenglei@bupt.edu.cn; anwannian2021@bupt.edu.cn; dongchen@bupt.edu.cn; xuxiaodong@bupt.edu.cn).}
\thanks{Boxiao Han is with the Future Mobile Technology Lab China Mobile Research Institude(e-mail: hanboxiao@chinamobile.com)}}

\markboth{}
{}
%{Shell \MakeLowercase{\textit{et al.}}: Bare Demo of IEEEtran.cls for Journals}

\maketitle

\begin{abstract}
In this research, we present a cooperative communication network based on two energy-harvesting (EH) decode-and-forward (DF) relays that contain the harvest-store-use (HSU) architecture and may harvest energy from the surrounding environment by utilizing energy buffers. To improve the performance of this network, an opportunistic Routing (OR) algorithm is presented that takes into account channel status information, relay position, and energy buffer status, as well as using the maximal ratio combining (MRC) at the destination to combine the signals received from the source and relays. The theoretical expressions for limiting distribution of energy stored in infinite-size buffers are derived from the discrete-time continuous-state space Markov chain model (DCSMC). In addition, the theoretical expressions for outage probability, throughput, and per-packet timeslot cost in the network are obtained utilizing both the limiting distributions of energy buffers and the probabilities of transmitter candidates set. Through numerous simulation results, it is demonstrated that simulation results match with corresponding theoretical results.
\end{abstract}


\IEEEpeerreviewmaketitle

\section{Introduction}

In recent years, due to the emphasis on green communications, there has been a great deal of interest in EH technique, which can harvest energy from the surrounding ambient source such as radio frequency energy, light energy, thermal energy, etc\cite{article1}. In addition to the advantages of energy saving, EH can also extend battery lifetime compared to the source charging the battery, which is helpful in the complicated case of replacing the battery\cite{article2}. For these reasons, numerous articles emphasize the combination of EH and communications.
\par In \cite{3}, a relay selection approach for the EH wireless body area network that takes into account both channel state information (CSI) and the energy condition of the relay buffer is proposed. For the sake of obtaining a tradeoff between the transmission energy and decoding energy, the power splitting protocol is employed at the EH relay of the EH-based cooperative communication network in \cite{4} and \cite{5}. In these papers, the harvest-use (HU) architecture is employed to manage energy, meaning that the node will instantly use the gathered energy.
	\par On the other hand, the harvest-store-use (HSU) architecture, which is widely explored as an efficient technique of energy harvesting from the environment, is comprised of three energy management steps: energy harvesting, energy storage, and energy use. In \cite{article3}, a simple communication system with EH using HSU architecture is proposed, assuming that the EH node harvests energy from a radio frequency signal broadcast by an access point in the downlink and uses the stored energy to transmit data in the uplink. Unlike \cite{article3}, which only investigated the on-off policy, both the best-effort policy and the on-off policy are analyzed in \cite{article4}. Both \cite{article3} and \cite{article4} derived the limiting distribution of stored energy in the energy buffer at the EH node by using the theory of discrete-time continuous-state Markov chains. The difference between the on-off policy and the best-effort policy is whether set target energy to transmit information or not. In the on-off policy, only when target M units of energy are available in the energy buffer, the EH node can transmit information to the destination at the cost of M units of energy. In the best-effort policy, no matter how much energy is in the energy buffer, the EH node uses all energy to transmit information to the destination. To maximize the end-to-end system throughput under data and energy storage constraints, the online and offline optimization algorithms for joint relay selection and power control have been discussed in \cite{article5}.
	\par In the aforementioned papers, all EH nodes harvest energy from RF signals received from the downlink, which is not efficient when the path loss in practical communication links is high. And they only consider the relay link without a direct link. The self-sustaining node (SSN), which harvests energy from the ambiance, is regarded as a solution to the high path loss link in the EH network. Self-sustaining node and direct link are covered in \cite{article6}, where the limiting distributions of energy for both the incremental on-off policy and the incremental best-effort policy are derived by using a discrete-time continuous-state space Markov chain. In addition, the expressions for outage probability and throughput are obtained. In this process, the maximal ratio combining technique is used to improve the performance of the network. In all of the papers mentioned above, each source node in the cooperative relay network is powered by the power source. For the sake of further energy saving, the source node is equipped with the EH technique in \cite{article7}. A decoded-and-forward energy-harvesting system with multiple relays in the presence of transmitting hardware impairments is proposed in \cite{article8}. The system has two transmitting policies. For policy-1, the selection of the best relay that assists the source-to-destination communication is based on the first-hop channel-state-information. For policy-2, the selection of the best relay is based on the second-hop CSI. 
	\par All relay networks mentioned above are primarily founded on wireless two-hop relay networks with a single relay or multiple relays with the same priority. However, actually, the practical wireless cooperative network like a wireless sensor network\cite{article9} or an Internet of Things (IoT) network\cite{article10} is the wireless multi-hop network. Opportunistic routing as a solution to improve the performance of the wireless multi-top network has been widely studied\cite{article11}. An energy-efficient OR, which exploits cross-layer information exchange and Energy-consumption-based Objective Functions, is proposed in \cite{article12} and \cite{article13} where it has been clarified that OR outperforms the traditional routing with the specific preselected route in terms of system performance. And in \cite{anwannian}, with using OR algorithm, a cooperative communication network based on two energy-harvesting decode-and-forward relays is proposed. 
	\begin{table}[h!]
		\caption{ Comparsion of references}
		\label{tab:2} 
		\begin{tabular}{|c|c|c|c|c|}
			\hline
			Reference &  EH buffer architecture & SSN &OR&MRC \\
			\hline
			\cite{3},\cite{4},\cite{5}&HU&N&N&N\\
			\hline
			\cite{article3},\cite{article4},\cite{article5}&HSU&N&N&N\\
			\hline
			\cite{article6},\cite{article7}&HSU&Y&N&Y\\
			\hline
			\cite{article8}&HSU&Y&N&N\\
			\hline
			\cite{article11}--\cite{article13}&No EH&N&Y&N\\
			\hline
			\cite{anwannian}&HSU&Y&Y&N\\
			\hline
			\cite{article14}&No EH&N&N&Y\\
			\hline
			Our work&HSU&Y&Y&Y\\
			\hline
		\end{tabular}\centering
	\end{table}
	\par According to the above researches, maximal-ratio combining as a well-studied technique\cite{article14}, which is helpful to improve performance in communication networks, is worth being combined with EH. This article is targeted at proposing and designing an analytical framework for the study of OR-aided cooperative communication networks with EH nodes and using MRC in our system. Table \ref{tab:2} shows the comparison between our work and the above references, where N indicates that the technology is not employed in the studied system, and Y indicates that the technology is employed in the studied system. In addition, we model energy buffers using discrete-time continuous-state space Markov chains. The contributions of this work are as follows:

\begin{enumerate}
  \item A cooperative network consisting of two EH DF relay nodes and employing the MRC and OR algorithms is proposed.
  \item The expressions for the limiting distributions of energy stored in buffers are obtained with the help of the DCSMC model and the state-transition-matrix-based (STM)  theoretical solutions of the probability of transmitter candidates set. The limiting distributions of energy are the base of further analysis.
  \item The closed-form expressions of outage probability and throughput are derived.
  \item Through numerous simulation results, it is demonstrated that simulation results match with corresponding theoretical results.
  
\end{enumerate}
\par The remaining paper is developed as follows:
Section II interprets the system model and OR protocol. In section III, the limiting PDFs of the energy stored in the buffers R1 and R2 are presented. In section IV, the expressions for outage probability and throughput are determined. In Section V, simulation performance results are displayed. Section VI is where the conclusion is provided.

$Notations:$

\begin{table}[h!]
	\caption{Notations table}
	\label{tab:1} 
	\begin{tabular}{|c|p{5.8cm}<{\centering}|}
		\hline
		Notation & Meaning \\
		\hline
		$A\sim\mathcal{CN}(0, \theta)$& the random variable $A$ follows the complex Gaussian distribution with mean 0 and variance $\theta$\\
		\hline
		$|B|$& the absolute value of $B$\\
		\hline
		$\mathbb{E}[\cdot]$&the expectation operator\\
		\hline
		$C_1 \cup C_2$& the union of the condition $C_1$ and the condition $C_2$\\
		\hline
		$C_1 \cap C_2$& the intersection of the condition $C_1$ and the condition $C_2$\\
		\hline
		$W(\cdot)$&the Lambert W function\\
		\hline
		Boldface capital $\textbf{T}$ & matrices or vectors\\
		\hline
		$\overline{\mathbb{C}}$&the comlementary event of $\mathbb{C}$\\
		\hline
		$*$&multiplication\\
		\hline
		\multirow{3}{*}{$|$}&conditional probability. For example, case A,B|C means event A and B occur under the condition of C.\\
		\hline
		\multirow{2}{*}{conv$(A,B)$}&discrete convolution $y(j)=\mathrm{conv}(A,B)=\sum_{i=-\infty}^{+\infty}A(i)B(j-i)$\\
		\hline
	\end{tabular}
\end{table}
In TABLE \ref{tab:1}, all notations are showed. \\

$Abbreviation:$

\begin{table}[h!]
	\caption{Abbreviations table}
	\label{tab:3} 
	\begin{tabular}{|c|p{5.8cm}<{\centering}|}
		\hline
		Abbreviation & Meaning \\
		\hline
		EH& energy harvesting\\
		\hline
		DF& decode and forward\\
		\hline
		HSU&harvest-store-use\\
		\hline
		OR& opportunistic routing\\
		\hline
		MRC& maximal ratio combing\\
		\hline
		DCSMC&discrete-time continuous-state space Markov chain\\
		\hline
		CSI& channel state information\\
		\hline
		HU&harvest-use\\
		\hline
		SSN&self-sustaining node\\
		\hline
		PEB&primary energy buffer \\
		\hline
		SEB&secondary energy buffer \\
		\hline
		ACK&positive acknowledgement \\
		\hline
		NACK&negative acknowledgement \\
		\hline
		TC&transmitter candidate \\
		\hline
	\end{tabular}
\end{table}
In TABLE \ref{tab:3}, all abbreviations are showed. 

%$A\sim\mathcal{CN}(0, \theta)$ indicates the random variable $A$ follows the complex Gaussian distribution with mean 0 and variance $\theta$. The absolute value of $B$ is denoted by $|B|$. $\mathbb{E}[\cdot]$ denotes the expectation operator. $C_1 \cup C_2$ denotes the union of the condition $C_1$ and the condition $C_2$, while $C_1 \cap C_2$ means the intersection of the condition $C_1$ and the condition $C_2$. $W(\cdot)$ is the Lambert W function. Boldface capital and lower-case letters stand for matrices and vectors, respectively. Term $\overline{\mathbb{C}}$ represents the comlementary event of $\mathbb{C}$. The sign '$*$' means multiplication.

\section{System Model and OR Protocol}
\subsection{System Model}
As shown in Fig. \ref{system}, the network proposed in this paper is composed of a source node S, two DF relay nodes R1 and R2, and a destination node D. All nodes are assumed as half-duplex nodes. Only the source and destination nodes are powered by the power supply, while relay nodes are equipped with energy buffers using HSU architectures to harvest ambient energy. In each time slot, only one node is activated and can transmit signals by broadcast way, while others are silent or receive the signals. Quasi-static Rayleigh fading is assumed between all nodes. Let $d_{ab}$ denote the distance between nodes with a and b. Hence, denoted by $h_{SD}(i)\sim\mathcal{CN}(0, d_{SD}^{\alpha})$, $h_{SR1}(i)\sim\mathcal{CN}(0, d_{SR1}^{\alpha})$, $h_{SR2}(i)\sim\mathcal{CN}(0, d_{SR2}^{\alpha})$, $h_{R1R2}(i)\sim\mathcal{CN}(0, d_{R1R2}^{\alpha})$, $h_{R1D}(i)\sim\mathcal{CN}(0, d_{R1D}^{\alpha})$ and $h_{R2D}(i)\sim\mathcal{CN}(0, d_{R2D}^{\alpha})$ the channel coefficients within the $i$-th time slot between S and D, S and R1, S and R2, R1 and R2, R1 and D, and R2 and D respectively, where the term $\alpha$ represents the path-loss exponent. 
\par The HSU architectures equipped by the relay nodes mainly have two energy buffers respectively, which are an infinite-size primary energy buffer (PEB) and an infinite-size secondary energy buffer (SEB), for the reason that storage devices are disabled to charge and discharge at the same time \cite{article15} \cite{article16}. The energy harvested from ambiance is stored in SEB first. Then, in order to make the energy buffer equipped by relay nodes R1 and R2 charge and discharge simultaneously, SEB transfers the energy to the PEB at the end of each timeslot. Finally, the energy in PEB is used to power the node to transmit signals. Due to the fact that the energy harvesting power is far smaller than the actual energy buffer capacity, the capacity of the energy buffer is considered to be infinite size in the theoretical analysis. In \cite{article16}, it is indicated that there is no energy loss during charging between SEB and PEB or discharging between PEB and the node transmitter.
The signals received from S, R1, and R2 are combined using MRC at D.

\begin{figure}[h]
	
	\centerline{\includegraphics[width=3.5in]{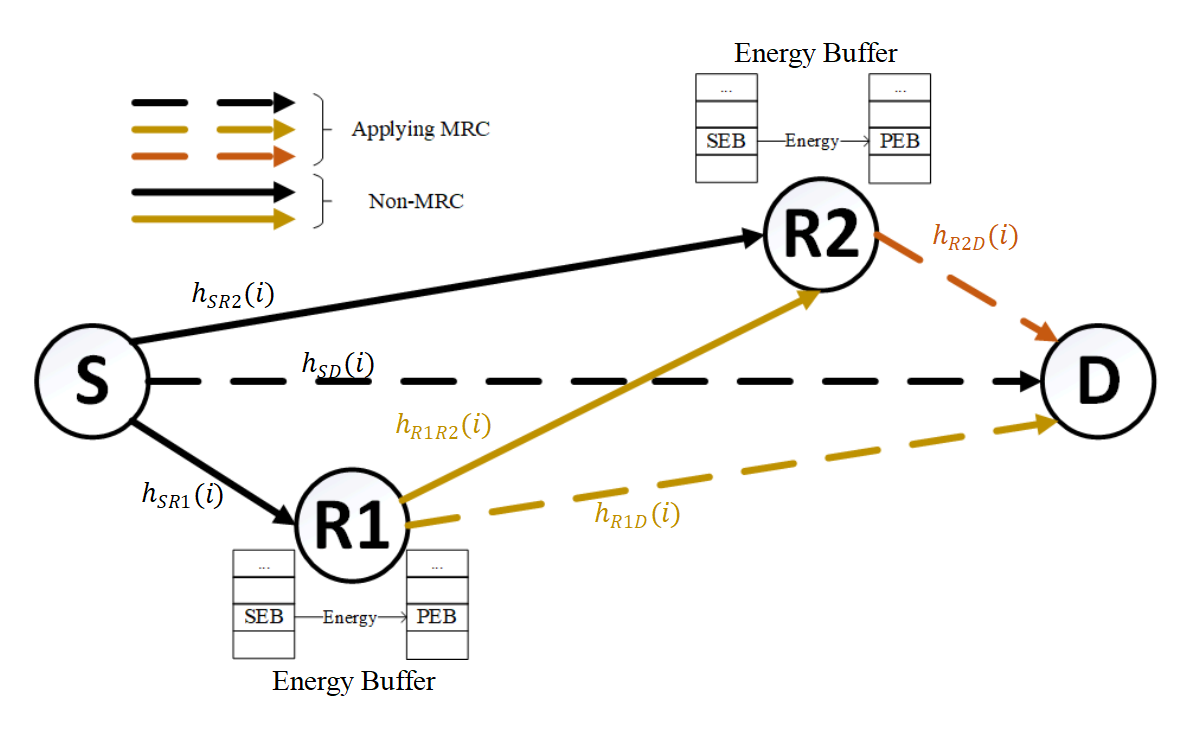}}
	\caption{System model\label{system}}
\end{figure}

\par The source S broadcasts unit-energy symbols $x_S(i)$ to relay R1, relay R2, and destination D at rate $R_0$ with the constant power $P_S$. The received signals $y_{SR1}(i)$, $y_{SR2}(i)$ and $y_{SD}(i)$ at relay R1, relay R2 and destination D in the $i$-th time slot are given by
\begin{equation}\label{eq1}
	y_{SR1}(i) = \sqrt{P_{S}}h_{SR1}(i)x_S(i) + n_{SR1}(i),
\end{equation}
\begin{equation}
	y_{SR2}(i) = \sqrt{P_{S}}h_{SR2}(i)x_S(i) + n_{SR2}(i),
\end{equation}
\begin{equation}
	y_{SD}(i) = \sqrt{P_{S}}h_{SD}(i)x_S(i) + n_{SD}(i),
\end{equation}
where, $n_{SR1}(i)$, $n_{SR2}(i)$ and $n_{SD}(i)\sim\mathcal{CN}(0, N_0)$ represent the received additive white Gaussian noise (AWGN) at R1, R2 and D, respectively. Thus, the received signal-to-noise ratios (SNRs) $\gamma_{SR1}(i)$, $\gamma_{SR2}(i)$ and $\gamma_{SD}(i)$ at R1, R2 and D in the $i$-th time slot would be given as follows
\begin{equation}
	\begin{split} %使用\split可以让多行公式对齐, 而且只有一个标号
		& \gamma_{SR1}(i) = \frac{P_S|h_{SR1}(i)|^2}{N_0}, \quad \gamma_{SR2}(i) = \frac{P_S|h_{SR2}(i)|^2}{N_0} \\  %&放置在最左边就是表示公式左对齐
		& \mbox{ and } \quad \gamma_{SD}(i) = \frac{P_S|h_{SD}(i)|^2}{N_0}.   %\mbox{ and }数学公式中写入非数学变量的字符,使用Roman字体，而不是数学变量的斜体
	\end{split}
\end{equation}

\par  As a DF relay, the relay R1 decodes the received signals, re-encodes them into unit-energy symbols $x_{R1}(i)$, and then broadcasts $x_{R1}(i)$ to relay R2 and destination D with the constant power $M_{R1}$. Similarly, the relay R2 obtains the unit-energy symbols $x_{R2}(i)$ in the same way as relay R1, and then transmits $x_{R2}(i)$ to the destination D with the constant power $M_{R2}$. Next, the received signals $y_{R1R2}(i)$, $y_{R1D}(i)$ and $y_{R2D}(i)$ at relay R2 and destination D in the $i$-th time slot can be represented by
\begin{equation}
	y_{R1R2}(i) = \sqrt{M_{R1}}h_{R1R2}(i)x_{R1}(i) + n_{R1R2}(i),
\end{equation}
\begin{equation}
	y_{R1D}(i) = \sqrt{M_{R1}}h_{R1D}(i)x_{R1}(i) + n_{R1D}(i),
\end{equation}
\begin{equation}
	y_{R2D}(i) = \sqrt{M_{R2}}h_{R2D}(i)x_{R2}(i) + n_{R2D}(i),
\end{equation}
where, $n_{R1R2}(i)$, $n_{R1D}(i)$ and $n_{R2D}(i)\sim\mathcal{CN}(0, N_0)$ denote the received AWGN at R2 and D, respectively. Similarly, the received SNRs $\gamma_{R1R2}(i)$, $\gamma_{R1D}(i)$ and $\gamma_{R2D}(i)$ at R2 and D can be expressed as follows
\begin{equation}
	\begin{split}
		&\gamma_{R1R2}(i) = \frac{M_{R1}|h_{R1R2}(i)|^2}{N_0}, \quad \gamma_{R1D}(i) = \frac{M_{R1}|h_{R1D}(i)|^2}{N_0} \\
		&\mbox{ and } \quad \gamma_{R2D}(i) = \frac{M_{R2}|h_{R2D}(i)|^2}{N_0}.
	\end{split}
\end{equation}
According to channel assumption, the PDF of SNR of channel SR1, SR2, SD, R1D, R1R2 and R2D can be expressed as 
$f_{Node1Node2}(x) = W_{Node1Node2} e^{-W_{Node1Node2} x}$.\\
For example, the PDF of SNR of channel SR1 can be expressed as
$f_{SR1}(x) = W_{SR1} e^{-W_{SR1} x}$, where, $W_{SD}=\frac{d_{SD}^{\alpha}N_0}{P_S}$. Similarly,  $W_{SR2}=\frac{d_{SR2}^{\alpha}N_0}{P_S}$, $W_{SR1}=\frac{d_{SR1}^{\alpha}N_0}{P_S}$, $W_{R1D}=\frac{d_{R1D}^{\alpha}N_0}{M_{R1}}$, $W_{R1R2}=\frac{d_{R1R2}^{\alpha}N_0}{M_{R1}}$,  $W_{R2D}=\frac{d_{R2D}^{\alpha}N_0}{M_{R2}}$.

\begin{figure}[h]
	\centerline{\includegraphics[width=3.5in]{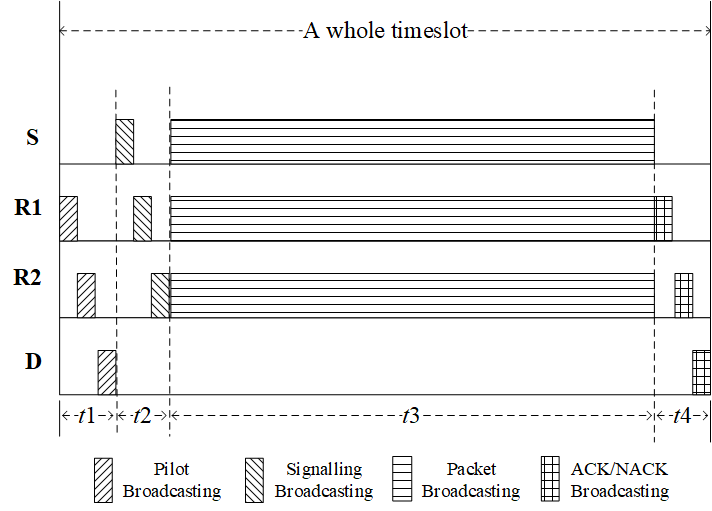}}
	\caption{Timeslot model\label{timeslot model}}
\end{figure}

Let 
\begin{equation}
	\Gamma_{th} = 2^{R_0}-1
\end{equation}
be the network SNR threshold, where $R_0$ represents the target data transmitting rate.

\par In our system, a whole timeslot is divided into four sub slots as shown in Fig. \ref{timeslot model}, where a whole time slot consists of pilot broadcasting sub-slot $t$1, signalling broadcasting sub-slot $t$2, packet broadcasting sub-slot $t$3 and positive acknowledgement (ACK) or negative acknowledgement (NACK) broadcasting sub-slot $t$4. In sub-slot $t$1, node D and relay nodes R1 and R2 need to broadcast the pilot signals sequentially, allowing each node to determine the CSI between itself and other nodes. Moreover, the pilot signal of D carries the overall SNR $\gamma_{overall}$ of the current received packet combined by the MRC technique so that S, R1, and R2 nodes can know the current SNR threshold $\Gamma_{th}-\gamma_{overall}$. About $\gamma_{overall}$ is descibed in OR protocol. In sub-slot $t$2, each node determines whether it is able to transmit the packet depending on its energy state, CSI and ownership of the packet. If the node can transmit the packet, then it broadcasts high-level signalling; otherwise, it broadcasts low-level signalling. And the order of S, R1, and R2 signalling broadcasting is shown in Fig. \ref{timeslot model}. In this way, each node can obtain the transmission state of the cooperative network. In sub-slot $t$3, with the transmission state, nodes determine whether transmit the packet or not based on the OR protocol described below. Then only one node can broadcast the packet in sub-slot $t$3. The blocks of sub-slot $t$3 shown in Fig. \ref{timeslot model} just means that S, R1, and R2 are all able to broadcast packet in sub-slot $t$3, but do not mean that they can broadcast packet at the same timeslot. In sub-slot $t$4, D, R1, and R2 send a positive acknowledgment or a negative acknowledgment to indicate their reception statuses. Only when the instantaneous link SNR at the receiving node is not less than the threshold $\Gamma_{th}$, ACK signals would be broadcasted by the receiving node as the sign of successful packet reception.

\subsection{OR Protocol}
In this part, the main steps implemented by the OR protocol are described. S, R1, and R2 are broadcast nodes. There can be more than one broadcast node as transmitter candidates (TC) at the same time which have the packet. TC set at the $i$-th timeslot is denoted by TCs$(i)$. The node which can receive the signals transmitted by the broadcast node is named the neighboring node of the broadcast node. Thus, in this system, neighboring nodes of S are R1, R2, and D, while neighboring nodes of R1 are R2 and D, and the neighboring node of R2 is D. Then, we assumed each broadcast node knows the perfect CSI of communication links between itself and each of its neighboring nodes by pilot broadcasting. Our cooperative network also assumes that S always has the packet to be transmitted. And if the signals with strong enough SNR, which is more than the current threshold, are transmitted to D, S will be set as the next broadcast node to transmit the new signals, and relay nodes clear data cache, which is indicated that one delivery is completed. Additionally, when transmitting the same packet to the same node, S has the highest transmission priority, R2 has the second priority, and R1 has the lowest priority. This is because S has fixed power support, which does not consider the energy state when transmitting packets. This way can save the harvested energy of relay nodes in order to be used in the case where only relay nodes can transmit the packet. Furthermore, the OR protocol consists of three steps:
\subsubsection{}
Determine the TCs $\textbf{S}(i) \in\{\textbf{s}_1, \textbf{s}_2, \textbf{s}_3, \textbf{s}_4,  \textbf{s}_5, \textbf{s}_6\}$ in the current time slot, where $\textbf{s}_1 = \{S\}$, $\textbf{s}_2 = \{S,R1\}$, $\textbf{s}_3 = \{S,R2\}$, $\textbf{s}_4 = \{(S,R1,R2)_1\}$, $\textbf{s}_5 = \{(S,R1,R2)_2\}$, $\textbf{s}_6 = \{(S,R1,R2)_3\}$. Especially, the term (S, R1, R2)$_1$ means the situation of the packet from R1 to R2, when both S and R1  have the packet. Term (S, R1, R2)$_2$ represents that the packet at R2 is from S, when both S and R1 have the packet. Term (S,R1,R2)$_3$ represents that S transmits the packet to R1 and R2 in the same slot.
\subsubsection{}
On the basis of the CSI between nodes, the stored energy $B_1(i)$ of $R1$ and the stored energy $B_2(i)$ of $R2$, the transmitter node $\in\{S, R1, R2, \approx\}$ is selected from the TCs $\textbf{S}(i)$, where term $\approx$ means that all nodes keep silent in the current time slot.
\subsubsection{}
Determine the effective transmission and get TCs $\textbf{S}(i+1) \in\{\textbf{s}_1, \textbf{s}_2, \textbf{s}_3, \textbf{s}_4,  \textbf{s}_5, \textbf{s}_6\}$ in the next time slot.

In details, all cases of OR protocol are shown in \ref{TableB}. Because of transmitting packets by broadcasting, D can receive all signals of the same packet to combine by MRC even in the different timeslot, no matter which the transmitter is. Therefore, we use $\gamma$ (like $\gamma_{SD}$) without $(i)$ to represent the SNR of the signal of the same packet received by D in the previous timeslot. In addition, all $\gamma(i)$ (like $\gamma_{SD}(i)$) represent the SNR of the corresponding link in the current timeslot. Assume that D received a signal with $\gamma_1$ SNR and a signal of the same packet with $\gamma_2$ SNR in the previous timeslot. After MRC technique processing, the overall SNR of the signal received by D is $\gamma_{overall}=\gamma_1+\gamma_2$. Hence, to meet the SNR threshold $\Gamma_{th}$, each node only needs to transmit the signal with $\Gamma_{th}-\gamma_{overall}$ SNR, which reduces the need for high-quality channels and improves the performance of the network. For example, when TCs$(i)$ = \{S, R1\}, the $\gamma_{overall}$ is $\gamma_{SD}$ which was received by D when S transmitted the signal to R1 in the previous timeslot by broadcasting way. Terms $\gamma_{overall1}$, $\gamma_{overall2}$ and $\gamma_{overall3}$ represent the SNR, which has been received by D, when TCs$(i)$ = \{S, R1\}, TCs$(i)$ = \{S, R2\} and TCs$(i)$ = \{S, R1, R2\} respectively. Obviously, $\gamma_{overall1}$, $\gamma_{overall2}$ and $\gamma_{overall3}$ are less than $\Gamma_{th}$.
\begin{table*}[h!]
	\caption{The OR Protocol}
In the current time slot, the protocol determines the transmitter node and effective transmission according to the known transmitter candidates set TCs(i) and the known conditions (energy information stored by relay nodes, the current channel information, and the overall SNR of the current packet which has been combined by MRC technique), and further obtains the TCs(i + 1) of the next time slot.
	\begin{tabular}{|c|c|c|c|c|}
		\hline
		$TCs(i)$ & transmitter & effective transmission & conditions & $TCs(i+1)$\\
		\hline
		\multirow{5}{*}{S}&S&S→D&$\gamma_{SD}(i)\ge\Gamma_{th}$&{S}\\
		\cline{2-5}
		&S&S→R1,R2&$\gamma_{SR1}(i)\ge\Gamma_{th}$,$\gamma_{SR2}(i)\ge\Gamma_{th},\gamma_{SD}(i)<\Gamma_{th}$&(S,R1,R2)$_3$ \\
		\cline{2-5}
		&S&S→R2&$\gamma_{SR2}(i)\ge\Gamma_{th},\gamma_{SR1}(i)<\Gamma_{th},\gamma_{SD}(i)<\Gamma_{th} $&S,R2 \\
		\cline{2-5}
		&S&S→R1&$\gamma_{SR1}(i)\ge\Gamma_{th},\gamma_{SR2}(i)<\Gamma_{th},\gamma_{SD}(i)<\Gamma_{th}$&S,R1 \\
		\cline{2-5}
		&$\approx$&$\approx$&others&S \\
		\hline
		\multirow{8}{*}{S,R1}&S&S→D&$\gamma_{overall1}+\gamma_{SD}(i)\ge\Gamma_{th}$&S\\
		\cline{2-5}
		&R1&R1→D&$\gamma_{overall1}+\gamma_{R1D}(i)\ge\Gamma_{th},B_1(i)\ge M_{R1},\gamma_{overall1}+\gamma_{SD}(i)<\Gamma_{th}$&S\\
		\cline{2-5}		
		&\multirow{2}{*}{S}&\multirow{2}{*}{S→R2}&$\gamma_{SR2}(i)\ge\Gamma_{th},\gamma_{overall1}+\gamma_{R1D}(i)<\Gamma_{th},B_1(i)\ge M_{R1},$&\multirow{2}{*}{(S,R1,R2)$_1$}\\
		&&&$\gamma_{overall1}+\gamma_{SD}(i)<\Gamma_{th}$&\\
		\cline{2-5}
		&S&S→R2&$\gamma_{SR2}(i)\ge\Gamma_{th},B_1(i)< M_{R1},\gamma_{overall1}+\gamma_{SD}(i)<\Gamma_{th}$&(S,R1,R2)$_1$\\
		\cline{2-5}
		&\multirow{2}{*}{R1}&\multirow{2}{*}{R1→R2}&$\gamma_{R1R2}(i)\ge\Gamma_{th},\gamma_{SR2}(i)<\Gamma_{th},B_1(i)\ge M_{R1},$&\multirow{2}{*}{(S,R1,R2)$_2$}\\
		&&&$\gamma_{overall1}+\gamma_{SD}(i)<\Gamma_{th}$&\\
		\cline{2-5}
		&$\approx$&$\approx$&others&S,R1\\
		\hline
		\multirow{3}{*}{S,R2}&S&S→D&$\gamma_{overall2}+\gamma_{SD}(i)\ge\Gamma_{th}$&S\\
		\cline{2-5}
		&R2&R2→D&$\gamma_{overall2}+\gamma_{R2D}(i)\ge\Gamma_{th},\gamma_{overall2}+\gamma_{SD}(i)<\Gamma_{th}$&S\\
		\cline{2-5}
		&$\approx$&$\approx$&others&S,R2\\
		\hline
		\multirow{6}{*}{(S,R1,R2)$_1$}&S&S→D&$\gamma_{overall3}+\gamma_{SD2}(i)\ge\Gamma_{th},B_2(i)\ge M_{R2}$&S\\
		\cline{2-5}
		&\multirow{2}{*}{R1}&\multirow{2}{*}{R1→D}&$\gamma_{overall3}+\gamma_{R1D}(i)\ge\Gamma_{th},\gamma_{overall3}+\gamma_{R2D}(i)<\Gamma_{th},$&\multirow{2}{*}{S}\\
		&&&$\gamma_{overall3}+\gamma_{SD}(i)<\Gamma_{th},B_1(i)\ge M_{R1}$&\\
		\cline{2-5}
		&\multirow{2}{*}{R2}&\multirow{2}{*}{R2→D}&$\gamma_{SD}+\gamma_{R1D}+\gamma_{R2D}(i)\ge\Gamma_{th},$&\multirow{2}{*}{S}\\
		&&&$\gamma_{overall3}+\gamma_{SD}(i)<\Gamma_{th},B_2(i)\ge M_{R2}$&\\
		\cline{2-5}
		&$\approx$&$\approx$&others&(S,R1,R2)$_1$\\
		\hline
		\multirow{6}{*}{(S,R1,R2)$_2$}&S&S→D&$\gamma_{overall3}+\gamma_{SD}(i)\ge\Gamma_{th}$&S\\
		\cline{2-5}
		&\multirow{2}{*}{R1}&\multirow{2}{*}{R1→D}&$\gamma_{overall3}+\gamma_{R1D}(i)\ge\Gamma_{th},\gamma_{overall3}+\gamma_{R2D}(i)\ge\Gamma_{th},$&\multirow{2}{*}{S}\\
		&&&$\gamma_{overall3}+\gamma_{SD}(i)<\Gamma_{th}$&\\
		\cline{2-5}
		&R2&R2→D&$\gamma_{overall3}+\gamma_{R2D}(i)\ge\Gamma_{th},\gamma_{overall3}+\gamma_{SD}(i)<\Gamma_{th}$&S\\
		\cline{2-5}
		&$\approx$&$\approx$&others&(S,R1,R2)$_2$\\
		\hline
		\multirow{4}{*}{(S,R1,R2)$_3$}&S&S→D&$\gamma_{overall3}+\gamma_{SD}(i)\ge\Gamma_{th}$&S\\
		\cline{2-5}
		&\multirow{2}{*}{R1}&\multirow{2}{*}{R1→D}&$\gamma_{overall3}+\gamma_{R1D}(i)\ge\Gamma_{th},\gamma_{overall3}+\gamma_{R2D}(i)<\Gamma_{th},$&\multirow{2}{*}{S}\\
		&&&$\gamma_{overall3}+\gamma_{SD}(i)<\Gamma_{th}$&\\
		\cline{2-5}
		&R2&R2→D&$\gamma_{overall3}+\gamma_{R2D}(i)\ge\Gamma_{th},\gamma_{overall3}+\gamma_{SD}(i)<\Gamma_{th}$&S\\
		\cline{2-5}
		&$\approx$&$\approx$&others&(S,R1,R2)$_3$\\
		\hline
	\end{tabular}\centering
	
	\label{TableB}
\end{table*}

\section{Limiting Distribution of Energy}
In this section, the expressions for the limiting distributions of energy in the PEB of R1 and R2, assuming it to be of infinite size, are derived and analyzed. The input harvest energy $X(i)$ in the  $i$-th time slot is assumed to be an exponentially distributed random variable with probability density function (PDF) $f_X(x)$. With the mean  $\mathbb{E}[X(i)]=1/\lambda$, the PDF of $X(i)$ can be expressed as 
\begin{equation}
	f_X(x) = \lambda e^{-\lambda x}, \quad x>0.
\end{equation}
Obviously, the energy harvesting parameter $\lambda$ of R1 and R2 are independent. Thus, let $\lambda_1$ represent the parameter of R1, and $\lambda_2$ represent the parameter of R2.

In order to clearly describe the theoretical derivation process, letters '$\mathfrak{c}$', '$\mathfrak{d}$', '$\mathfrak{f}$', '$\mathfrak{g}$', '$\mathfrak{m}$', '$\mathfrak{n}$' and '$\mathfrak{o}$' are used, whose upper cases represent conditions and lower cases represent corresponding probability values. For example,
\begin{equation}\label{c}
	\begin{split}
\mathbb{C} &\triangleq \gamma_{overall1} + \gamma_{SD2}(i)< \Gamma_{th},
	\end{split}
\end{equation}
\begin{equation}
	\begin{split}
		&\mathrm{Pr}\{\gamma_{overall1} + \gamma_{SD}(i) < \Gamma_{th} \}=\mathrm{Pr}\{\mathbb{C}\}=\mathfrak{c}\\
	\end{split}
\end{equation}
 Case-$\mathbb{C}$ means the D node receives a signal from the S node first, but the SNR of the signal is less than the threshold. Furthermore, after some slots, D receives the signal from S again, while after the MRC algorithm combines the signals, the sum of SNR is still less than the threshold. And, obviously, $\overline{\mathbb{C}}$ means the sum of SNR is greater than or equal to the threshold as well as $\overline{\mathbb{C}}=1-\mathfrak{c}$.
	\begin{table*}[h!]
		\caption{All failure transmission case only considering CSI in system table}
	\begin{tabular}{|c|c|c|c|c|}
	\hline
	$TCs(I)$ & transmitter & recevier & condition & alias\\
	\hline
	\multirow{4}{*}{S}&S&D&$\gamma_{SD}(i)<\Gamma_{th}$&$\mathbb{A}$\\
	\cline{2-5}
	&S&R1,R2&$\gamma_{SR1}(i)<\Gamma_{th}$,$\gamma_{SR2}(i)<\Gamma_{th}$&$\mathbb{B}$ \\
		\cline{2-5}
	&S&R2&$\gamma_{SR2}(i)<\Gamma_{th}$&$\mathbb{Q}$ \\
		\cline{2-5}
	&S&R1&$\gamma_{SR1}(i)<\Gamma_{th}$&$\mathbb{R}$ \\
	\hline
	\multirow{4}{*}{S,R1}&S&D&$\gamma_{overall1}+\gamma_{SD}(i)<\Gamma_{th}$&$\mathbb{C}$\\
	\cline{2-5}
		&R1&D&$\gamma_{overall1}+\gamma_{R1D}(i)<\Gamma_{th}$&$\mathbb{D}$\\
		\cline{2-5}		
		&S&R2&$\gamma_{SR2}<\Gamma_{th}$&$\mathbb{S}$\\
		\cline{2-5}
		&R1&R2&$\gamma_{R1R2}<\Gamma_{th}$&$\mathbb{T}$\\
	\hline
	\multirow{2}{*}{S,R2}&S&D&$\gamma_{overall2}+\gamma_{SD}(i)<\Gamma_{th}< \Gamma_{th}$&$\mathbb{M}$\\
	\cline{2-5}
		&R2&D&$\gamma_{overall2}+\gamma_{R2D}(i)<\Gamma_{th}$&$\mathbb{N}$\\
	\hline
	\multirow{3}{*}{(S,R1,R2)$_1$}&S&D&$\gamma_{overall3}+\gamma_{SD}(i)<\Gamma_{th}$&$\mathbb{O}$\\
	\cline{2-5}
	&R1&D&$\gamma_{overall3}+\gamma_{R1D}(i)<\Gamma_{th}$&$\mathbb{G}$\\
	\cline{2-5}
	&R2&D&$\gamma_{overall3}+\gamma_{R2D}(i)<\Gamma_{th}$&$\mathbb{F}$\\
		\hline
	\multirow{3}{*}{(S,R1,R2)$_2$}&S&D&$\gamma_{overall3}+\gamma_{SD}(i)<\Gamma_{th}$&$\mathbb{O}$\\
	\cline{2-5}
	&R1&D&$\gamma_{overall3}+\gamma_{R1D}(i)<\Gamma_{th}$&$\mathbb{G}$\\
	\cline{2-5}
	&R2&D&$\gamma_{overall3}+\gamma_{R2D}(i)<\Gamma_{th}$&$\mathbb{F}$\\
	\hline
	\multirow{3}{*}{(S,R1,R2)$_3$}&S&D&$\gamma_{overall3}+\gamma_{SD}(i)<\Gamma_{th}$&$\mathbb{O}$\\
	\cline{2-5}
	&R1&D&$\gamma_{overall3}+\gamma_{R1D}(i)<\Gamma_{th}$&$\mathbb{G}$\\
	\cline{2-5}
	&R2&D&$\gamma_{overall3}+\gamma_{R2D}(i)<\Gamma_{th}$&$\mathbb{F}$\\
	\hline
	\end{tabular}\centering

\label{TableA}
	\end{table*}
In TABLE \ref{TableA}, all alias are shown.

Then, we can similarly get
\begin{equation}
	\begin{split}
		\mathbb{D} &\triangleq \gamma_{overall1} + \gamma_{R1D}(i)< \Gamma_{th},
	\end{split}
\end{equation}
\begin{equation}
	\begin{split}
		\mathbb{F} &\triangleq \gamma_{overall3} + \gamma_{R2D}(i)< \Gamma_{th},
	\end{split}
\end{equation}
\begin{equation}
	\begin{split}
		\mathbb{G} &\triangleq \gamma_{overall3} + \gamma_{R1D2}(i) < \Gamma_{th},
	\end{split}
\end{equation}
\begin{equation}
	\begin{split}
		\mathbb{M} &\triangleq \gamma_{overall2} + \gamma_{R2D}(i) < \Gamma_{th},
	\end{split}
\end{equation}
\begin{equation}
	\begin{split}
		\mathbb{N} &\triangleq \gamma_{overall2} + \gamma_{R2D}(i) < \Gamma_{th},
	\end{split}
\end{equation}
\begin{equation}
	\begin{split}
		\mathbb{O} &\triangleq \gamma_{overall3} + \gamma_{SD2}(i) < \Gamma_{th},
	\end{split}
\end{equation}

To obtain the expressions for PDF of energy in the PEB of R1 and R2, the probability of the transmitter candidate state needs to be calculated in advance. In \cite{anwannian}, we know, by using the state transition matrix (STM) $\textbf{T}$ of two adjacent time slots, each transmitter candidate state occurrence probabilities, $p_s = \mathrm{Pr}\{TCs = \{S\}\}$, $p_{SR1} = \mathrm{Pr}\{TCs = \{S,R1\}\}$, $p_{SR2} = \mathrm{Pr}\{TCs = \{S,R2\}\}$, $p_1= \mathrm{Pr}\{P1\}$, $p_2=\mathrm{Pr}\{P2\}$ and $p_3=\mathrm{Pr}\{P3\}$, can be updated iteratively until they converge to get their thoerical value, where $P1$ represents (S,R1,R2)$_1$, $P2$ represents (S,R1,R2)$_2$, $P3$ represents (S,R1,R2)$_3$. Obviously, $p_1+p_2+p_3=p_{SR1R2}=\mathrm{Pr}\{TCs = \{S,R1,R2\}\}$. And the STM $\textbf{T}$ can be constructed as follows
	\begin{strip}
\begin{equation}
	\textbf{T}=\begin{bmatrix}
		p_{S-S}      & p_{S-SR1}  & p_{S-SR2}  & p_{S-(SR1R2)_1}  & p_{S-(SR1R2)_2}   & p_{S-(SR1R2)_3}  \\
		p_{SR1-S}      & p_{SR1-SR1}  & p_{SR1-SR2}  & p_{SR1-(SR1R2)_1}  & p_{SR1-(SR1R2)_2}   & p_{SR1-(SR1R2)_3}  \\
		p_{SR2-S}      & p_{SR2-SR1}  & p_{SR2-SR2}  & p_{SR2-(SR1R2)_1}  & p_{SR2-(SR1R2)_2}   & p_{SR2-(SR1R2)_3}  \\
		p_{(SR1R2)_1-S}      & p_{(SR1R2)_1-SR1}  & p_{(SR1R2)_1-SR2}  & p_{(SR1R2)_1-(SR1R2)_1}  & p_{(SR1R2)_1-(SR1R2)_2}   & p_{(SR1R2)_1-(SR1R2)_3}  \\
		p_{(SR1R2)_2-S}      & p_{(SR1R2)_2-SR1}  & p_{(SR1R2)_2-SR2}  & p_{(SR1R2)_2-(SR1R2)_1}  & p_{(SR1R2)_2-(SR1R2)_2}   & p_{(SR1R2)_2-(SR1R2)_3}  \\
		p_{(SR1R2)_3-S}      & p_{(SR1R2)_3-SR1}  & p_{(SR1R2)_3-SR2}  & p_{(SR1R2)_3-(SR1R2)_1}  & p_{(SR1R2)_3-(SR1R2)_2}   & p_{(SR1R2)_3-(SR1R2)_3}  \\
	\end{bmatrix},
\end{equation}
	\end{strip}
where the element 

\noindent$p_{i-j}(i,j\in\{S, SR1, SR2, (SR1R2)_1, (SR1R2)_2, (SR1R2)_3\})$ represents the probability that the transmitter candidate state changes from state $i$ at the current time slot to state $j$ at the next time slot, satisfying $\begin{matrix} \sum_{j\in\{S, SR1, SR2, (SR1R2)_1, (SR1R2)_2, (SR1R2)_3\}} p_{i-j} \end{matrix}=1$. The details of the algorithm for the STM-based theoretical solutions of the transmitter candidate state probability can be seen in \cite{anwannian}. According to the OR protocol introduced in section II, we have
\begin{equation}
	\begin{split}
		p_{S-S} & = \mathrm{Pr}\{\gamma_{SD}(i)<\Gamma_{th}, \gamma_{SR2}(i)<\Gamma_{th}, \gamma_{SR1}(i)<\Gamma_{th}\} \\
		& \quad + \mathrm{Pr}\{\gamma_{SD}(i) \ge \Gamma_{th}\} \\
		& = \left(1-e^{-W_{SD}\Gamma_{th}}\right)\left(1-e^{-W_{SR2}\Gamma_{th}}\right)\left(1-e^{-W_{SR1}\Gamma_{th}}\right) \\
		& \quad + e^{-W_{SD}\Gamma_{th}},
	\end{split}
\label{P_s}
\end{equation}
\begin{equation}
	\begin{split}
		p_{S-SR1} & = \mathrm{Pr}\{\gamma_{SD}(i)<\Gamma_{th}, \gamma_{SR2}(i)<\Gamma_{th}, \gamma_{SR1}(i)\ge \Gamma_{th}\} \\
		& = \left(1-e^{-W_{SD}\Gamma_{th}}\right)\left(1-e^{-W_{SR2}\Gamma_{th}}\right)e^{-W_{SR1}\Gamma_{th}},
	\end{split}
\end{equation}
\begin{equation}
	\begin{split}
		p_{S-SR2} & = \mathrm{Pr}\{\gamma_{SD}(i)<\Gamma_{th}, \gamma_{SR1}(i)<\Gamma_{th}, \gamma_{SR2}(i)\ge \Gamma_{th}\} \\
		& = \left(1-e^{-W_{SD}\Gamma_{th}}\right)\left(1-e^{-W_{R1D}\Gamma_{th}}\right)e^{-W_{SR2}\Gamma_{th}},
	\end{split}
\end{equation}
\begin{equation}
	\begin{split}
		p_{S-(SR1R2)_3} & = \mathrm{Pr}\{\gamma_{SD}(i)<\Gamma_{th}, \gamma_{SR1}(i)\ge\Gamma_{th},\\
		&\qquad \gamma_{SR2}(i)\ge \Gamma_{th}\} \\
		& = \left(1-e^{-W_{SD}\Gamma_{th}}\right)e^{-W_{R1D}\Gamma_{th}}e^{-W_{SR2}\Gamma_{th}},
	\end{split}
\end{equation}
\begin{equation}
	\begin{split}
		p_{SR1-S} & = \mathrm{Pr}\{\overline{\mathbb{C}} \} + \mathrm{Pr}\{B_1(i) \ge M_{R1}, \mathbb{C}, \overline{\mathbb{D}}\} \\
		&=(1-\mathfrak{c})+\mathfrak{c}*PU1*(1-\mathfrak{d}),
	\end{split}
\end{equation}
\begin{equation}
	\begin{split}
		p_{SR1-(SR1R2)_1} & = \mathrm{Pr}\{\mathbb{C}, B_1(i) \ge M_{R1}, \mathbb{D}, \gamma_{R1R2}(i) \ge \Gamma_{th}, \\
		&\quad \gamma_{SR2}(i) \ge \Gamma_{th}\} \\
		&=\mathfrak{c}*PU1*\mathfrak{d}*e^{-W_{R1R2}\Gamma_{th}}*(1-^{-W_{SR2}\Gamma_{th}}),
	\end{split}
\end{equation}
\begin{equation}
	\begin{split}
		p_{SR1-(SR1R2)_2} & = \mathrm{Pr}\{\mathbb{C}, B_1(i) \ge M_{R1}, \mathbb{D}, \gamma_{SR2}(i) \ge \Gamma_{th}\} \\
		&+ \mathrm{Pr}\{\mathbb{C}, B_1(i) < M_{R1}, \gamma_{R1R2}(i) \ge \Gamma_{th}\}\\
		&=\mathfrak{c}*PU1*\mathfrak{d}*e^{-W_{SR2}\Gamma_{th}}+\mathfrak{c}*(1-PU1)\\
		&\times e^{-W_{SR2}\Gamma_{th}},
	\end{split}
\end{equation}
\begin{equation}
	\begin{split}
		p_{SR1-SR1} & =\mathrm{Pr}\{B_1(i) < M_{R1},\mathbb{C}\}\\
		 &+ \mathrm{Pr}\{B_1(i) \ge M_{R1}, \mathbb{C}, \mathbb{D}, \gamma_{R1R2}(i) \ge \Gamma_{th},\\
		 &\gamma_{SR2}(i) \ge \Gamma_{th}\}+\mathrm{Pr}\{B_1(i) \ge M_{R1}, \mathbb{C}, \mathbb{D},\\
		 & \gamma_{R1R2}(i) < \Gamma_{th}, \gamma_{SR2}(i) < \Gamma_{th}\}\\
		&=1-p_{SR1-(SR1R2)_1}-p_{SR1-S}-p_{SR1-(SR1R2)_2},
	\end{split}
\end{equation}
\begin{equation}
	\begin{split}
		p_{SR2-S} & =\mathrm{Pr}\{\overline{\mathbb{M}}\} + \mathrm{Pr}\{B_2(i) \ge M_{R2}, \mathbb{M},\overline{\mathbb{N}}\} \\
		&=(1-\mathfrak{m})+PU2*\mathfrak{m}*(1-\mathfrak{n}),
	\end{split}
\end{equation}
\begin{equation}
	\begin{split}
		p_{SR2-SR2} & =\mathrm{Pr}\{B_2(i) \ge M_{R2}, \mathbb{M}, \mathbb{N}\} + \mathrm{Pr}\{B_2(i) < M_{R2}\} \\
		&=1-p_{SR2-S},
	\end{split}
\end{equation}
\begin{equation}
	\begin{split}
		p_{(SR1R2)_1-S} & =\mathrm{Pr}\{\overline{\mathbb{O}}\} + \mathrm{Pr}\{\mathbb{O},B_2(i) \ge M_{R2},\overline{\mathbb{F}}\}\\
		&+\mathrm{Pr}\{\mathbb{O},B_2(i) \ge M_{R2},\mathbb{F},B_1(i) \ge M_{R1}, \overline{\mathbb{G}}\}\\
		&+\mathrm{Pr}\{\mathbb{O},B_2(i) < M_{R2},B_1(i) \ge M_{R1}, \overline{\mathbb{G}}\}\\
		&=(1-\mathfrak{o})+\mathfrak{o}*(PU2*(1-\mathfrak{f})\\
		&+(PU2*\mathfrak{f}+(1-PU2))*PU1*(1-\mathfrak{g})),
	\end{split}
\end{equation}
\begin{equation}
	\begin{split}
		p_{(SR1R2)_1-(SR1R2)_1} & =1-p_{(SR1R2)_1-S},
	\end{split}
\end{equation}
\begin{equation}
	\begin{split}
		p_{(SR1R2)_2-S} & =\mathrm{Pr}\{\overline{\mathbb{O}}\} + \mathrm{Pr}\{\mathbb{O},B_2(i) \ge M_{R2},\overline{\mathbb{F}}\}\\
		&+\mathrm{Pr}\{\mathbb{O},B_2(i) \ge M_{R2},\mathbb{F},B_1(i) \ge M_{R1}, \overline{\mathbb{G}}\}\\
		&+\mathrm{Pr}\{\mathbb{O},B_2(i) < M_{R2},B_1(i) \ge M_{R1}, \overline{\mathbb{G}}\}\\
		&=(1-\mathfrak{o})+\mathfrak{o}*(PU2*(1-\mathfrak{f})\\
		&+(PU2*\mathfrak{f}+(1-PU2))*PU1*(1-\mathfrak{g})),
	\end{split}
\end{equation}
\begin{equation}
	\begin{split}
		p_{(SR1R2)_2-(SR1R2)_2} & =1-p_{(SR1R2)_2-S},
	\end{split}
\end{equation}

~

\begin{equation}
	\begin{split}
		p_{(SR1R2)_3-S} & =\mathrm{Pr}\{\overline{\mathbb{O}}\} + \mathrm{Pr}\{\mathbb{O},B_2(i) \ge M_{R2},\overline{\mathbb{F}}\}\\
		&+\mathrm{Pr}\{\mathbb{O},B_2(i) \ge M_{R2},\mathbb{F},B_1(i) \ge M_{R1}, \overline{\mathbb{G}}\}\\
		&+\mathrm{Pr}\{\mathbb{O},B_2(i) < M_{R2},B_1(i) \ge M_{R1}, \overline{\mathbb{G}}\}\\
		&=(1-\mathfrak{o})+\mathfrak{o}*(PU2*(1-\mathfrak{f})\\
		&+(PU2*\mathfrak{f}+(1-PU2))*PU1*(1-\mathfrak{g})),
	\end{split}
\end{equation}
\begin{equation}
	\begin{split}
		p_{(SR1R2)_3-(SR1R2)_3} & =1-p_{(SR1R2)_3-S},
	\end{split}
\label{p_{(SR1R2)_3-(SR1R2)_3}}
\end{equation}
where, $PU1=\frac{1}{b_1\lambda_1 M_{R1}} $ and $PU2=\frac{1}{b_2\lambda_2 M_{R2}}$. $b_1$ and $ b_2$ are given in (\ref{b_1}) and  (\ref{b_2}). And the proof of $PU1$ and $PU2$ are given in (\ref{SEC IV7}) and (\ref{SEC IV18}). And the remaining elements of \textbf{T} are equal to 0.
\par In order to get \textbf{T}, the values of $c$, $\mathfrak{d}$, $\mathfrak{f}$, $g$, $m$, $n$ and $o$ are required. Before this step, it is foundation to get the probability distributions $p_{\gamma_{overall1}}(j)$, $p_{\gamma_{overall2}}(j)$, $p_{\gamma_{overall3}}(j)$ of $\gamma_{overall1}$, $\gamma_{overall2}$ and $\gamma_{overall3}$. For example, $p_{\gamma_{overall1}}(j)$ is the probability that node D has received the SNR of value $(\frac{j-1}{\mathcal{N}}\Gamma_{th}\leq \gamma_{overall1}<\frac{j}{\mathcal{N}}\Gamma_{th})$ when TCs = $\{S,R1\}$. To obtain these probability distributions, the STM method is used. STM $\textbf{T1}$, STM $\textbf{T2}$ and STM $\textbf{T3}$ correspond to $p_{\gamma_{overall1}}$, $p_{\gamma_{overall2}}$ and $p_{\gamma_{overall3}}$ respectively. The size of $\textbf{T1}$ is  $\mathcal{N}$ $\times$ $\mathcal{N}$ dimension. The states of $\textbf{T1}$ are the value of $\gamma_{overall1}$ and the element $\textbf{T1}(i,j)$ means the transition probability from $(\frac{i-1}{\mathcal{N}}\Gamma_{th}\leq \gamma_{overall1}<\frac{i}{\mathcal{N}}\Gamma_{th})$ to $(\frac{j-1}{\mathcal{N}}\Gamma_{th}\leq \gamma_{overall1}<\frac{j}{\mathcal{N}}\Gamma_{th})$, where $1\leq i,j\leq \mathcal{N}$ and $i,j\in Z$. The elements of $\textbf{T1}$ are as follows:
\begin{equation}
	\begin{split}
		\textbf{T1}(i,j) &=(1-\int_{0}^{\frac{\mathcal{N}-i}{\mathcal{N}}\Gamma_{th}}f_{SD}(x)dx\\
		&\times(PU1\int_{0}^{\frac{\mathcal{N}-i}{\mathcal{N}}\Gamma_{th}}f_{R1D}(x)dx\\
		&\times \int_{0}^{\Gamma_{th}}f_{SR2}(x)dx\int_{0}^{\Gamma_{th}}f_{R1R2}(x)dx\\
		&+(1-PU1)\int_{0}^{\Gamma_{th}}f_{SR2}(x)dx))\\
		&\times \int_{\frac{j-1}{\mathcal{N}}\Gamma_{th}}^{\frac{j}{\mathcal{N}}\Gamma_{th}}f_{SD}(x)dx/\int_{0}^{\Gamma_{th}}f_{SD}(x)dx\\
		& \qquad \qquad \qquad \qquad \qquad \qquad \qquad \quad i \neq j
	\end{split}\label{T1}
\end{equation}
\begin{equation}
	\begin{split}
		\textbf{T1}(i,j) &=(1-\int_{0}^{\frac{\mathcal{N}-i}{\mathcal{N}}\Gamma_{th}}f_{SD}(x)dx\\
		&\times(PU1\int_{0}^{\frac{\mathcal{N}-i}{\mathcal{N}}\Gamma_{th}}f_{R1D}(x)dx\\
		&\times \int_{0}^{\Gamma_{th}}f_{SR2}(x)dx\int_{0}^{\Gamma_{th}}f_{R1R2}(x)dx\\
		&+(1-PU1)\int_{0}^{\Gamma_{th}}f_{SR2}(x)dx))\\
		&\times \int_{\frac{j-1}{\mathcal{N}}\Gamma_{th}}^{\frac{j}{\mathcal{N}}\Gamma_{th}}f_{SD}(x)dx/\int_{0}^{\Gamma_{th}}f_{SD}(x)dx\\
		&+ \int_{0}^{\frac{\mathcal{N}-i}{\mathcal{N}}\Gamma_{th}}f_{SD}(x)dx\\
		&\times(PU1\int_{0}^{\frac{\mathcal{N}-i}{\mathcal{N}}\Gamma_{th}}f_{R1D}(x)dx\\
		&\times \int_{0}^{\Gamma_{th}}f_{SR2}(x)dx\int_{0}^{\Gamma_{th}}f_{R1R2}(x)dx\\
		&+(1-PU1)\int_{0}^{\Gamma_{th}}f_{SR2}(x)dx)\\
		& \qquad \qquad \qquad \qquad \qquad \qquad \qquad \quad i = j
	\end{split}
\end{equation}
%%integral(fun,0,Tth-m*Tth/m1)*(integral(fun1,0,Tth-m*Tth/m1)*integral(funSR2,0,Tth)*integral(funR1R2,0,Tth)*PB1+(1-PB1)*integral(funSR2,0,Tth))
%%integral(fun,0,Tth-m*Tth/m1)*(integral(fun1,0,Tth-m*Tth/m1)*integral(funSR2,0,Tth)*integral(funR1R2,0,Tth)*PB1+(1-PB1)*integral(funSR2,0,Tth)))*integral(fun,(n-1)*Tth/m1,n*Tth/m1)/integral(fun,0,Tth);
Similarly, the elements of $\textbf{T2}$ and $\textbf{T3}$ are as follows:
\begin{equation}
	\begin{split}
		\textbf{T2}(i,j) &=(1-\int_{0}^{\frac{\mathcal{N}-i}{\mathcal{N}}\Gamma_{th}}f_{SD}(x)dx\\
		&\times(PU2\int_{0}^{\frac{\mathcal{N}-i}{\mathcal{N}}\Gamma_{th}}f_{R2D}(x)dx+(1-PU2)))\\
		&\times \int_{\frac{j-1}{\mathcal{N}}\Gamma_{th}}^{\frac{j}{\mathcal{N}}\Gamma_{th}}f_{SD}(x)dx/\int_{0}^{\Gamma_{th}}f_{SD}(x)dx\\
		& \qquad \qquad \qquad \qquad \qquad \qquad \qquad \quad i \neq j
	\end{split}
\end{equation}
%%(1-integral(fun,0,Tth-m*Tth/m1)*(integral(fun2,0,Tth-m*Tth/m1)*PB+1-PB))*integral(fun,(n-1)*Tth/m1,n*Tth/m1)/integral(fun,0,Tth)
\begin{equation}
	\begin{split}
		\textbf{T2}(i,j) &=(1-\int_{0}^{\frac{\mathcal{N}-i}{\mathcal{N}}\Gamma_{th}}f_{SD}(x)dx\\
		&\times(PU2\int_{0}^{\frac{\mathcal{N}-i}{\mathcal{N}}\Gamma_{th}}f_{R2D}(x)dx+(1-PU2)))\\
		&\times \int_{\frac{j-1}{\mathcal{N}}\Gamma_{th}}^{\frac{j}{\mathcal{N}}\Gamma_{th}}f_{SD}(x)dx/\int_{0}^{\Gamma_{th}}f_{SD}(x)dx\\
		&+\int_{0}^{\frac{\mathcal{N}-i}{\mathcal{N}}\Gamma_{th}}f_{SD}(x)dx\\
		&\times(PU2\int_{0}^{\frac{\mathcal{N}-i}{\mathcal{N}}\Gamma_{th}}f_{R2D}(x)dx+(1-PU2))\\
		& \qquad \qquad \qquad \qquad \qquad \qquad \qquad \quad i = j
	\end{split}\label{T2}
\end{equation}

\begin{equation}
	\begin{split}
		\textbf{T3}(i,j) &=(1-\int_{0}^{\frac{\mathcal{N}-i}{\mathcal{N}}\Gamma_{th}}f_{SD}(x)dx\\
		&\times(PU2\int_{0}^{\frac{\mathcal{N}-i}{\mathcal{N}}\Gamma_{th}}f_{R2D}(x)dx+(1-PU2))\\
		&\times(PU1\int_{0}^{\frac{\mathcal{N}-i}{\mathcal{N}}\Gamma_{th}}f_{R1D}(x)dx+(1-PU1)))\\
		&\times (\frac{p_3}{p_1+p_2+p_3}\frac{\int_{\frac{j-1}{\mathcal{N}}\Gamma_{th}}^{\frac{j}{\mathcal{N}}\Gamma_{th}}f_{SD}(x)dx}{\int_{0}^{\Gamma_{th}}f_{SD}(x)dx}\\
		&+\frac{p_2}{p_1+p_2+p_3}\frac{p_{d_{S->R2}}(j)}{\sum_{k=1}^{\mathcal{N}}p_{d_{S->R2}}(k)}\\
		&+\frac{p_1}{p_1+p_2+p_3}\frac{p_{d_{R1->R2}}(j)}{\sum_{k=1}^{\mathcal{N}}p_{d_{R1->R2}}(k)})\\
		& \qquad \qquad \qquad \qquad \qquad \qquad \qquad \quad i \neq j
	\end{split}\label{T31}
\end{equation}
%%(1-integral(fun,0,Tth-m*Tth/m1)*(integral(fun2,0,Tth-m*Tth/m1)*PB+1-PB)*(integral(fun1,0,Tth-m*Tth/m1)*PB1+1-PB1))*(P3/(P1+P2+P3)*integral(fun,(n-1)*Tth/m1,n*Tth/m1)/integral(fun,0,Tth)+P2/(P1+P2+P3)*pd_SR2(n)+P1/(P1+P2+P3)*pd_R1R2(n));
\begin{equation}
	\begin{split}
		\textbf{T3}(i,j) &=(1-\int_{0}^{\frac{\mathcal{N}-i}{\mathcal{N}}\Gamma_{th}}f_{SD}(x)dx\\
		&\times(PU2\int_{0}^{\frac{\mathcal{N}-i}{\mathcal{N}}\Gamma_{th}}f_{R2D}(x)dx+(1-PU2))\\
		&\times(PU1\int_{0}^{\frac{\mathcal{N}-i}{\mathcal{N}}\Gamma_{th}}f_{R1D}(x)dx+(1-PU1)))\\
		&\times (\frac{p_3}{p_1+p_2+p_3}\frac{\int_{\frac{j-1}{\mathcal{N}}\Gamma_{th}}^{\frac{j}{\mathcal{N}}\Gamma_{th}}f_{SD}(x)dx}{\int_{0}^{\Gamma_{th}}f_{SD}(x)dx}\\
		&+\frac{p_2}{p_1+p_2+p_3}\frac{p_{d_{S->R2}}(j)}{\sum_{k=1}^{\mathcal{N}}p_{d_{S->R2}}(k)}\\
		&+\frac{p_1}{p_1+p_2+p_3}\frac{p_{d_{R1->R2}}(j)}{\sum_{k=1}^{\mathcal{N}}p_{d_{R1->R2}}(k)})\\
		&+\int_{0}^{\frac{\mathcal{N}-i}{\mathcal{N}}\Gamma_{th}}f_{SD}(x)dx\\
		&\times(PU2\int_{0}^{\frac{\mathcal{N}-i}{\mathcal{N}}\Gamma_{th}}f_{R2D}(x)dx+(1-PU2))\\
		&\times(PU1\int_{0}^{\frac{\mathcal{N}-i}{\mathcal{N}}\Gamma_{th}}f_{R1D}(x)dx+(1-PU1)))\\
		& \qquad \qquad \qquad \qquad \qquad \qquad \qquad \quad i = j
	\end{split}
\end{equation}
where $p_{d_{S->R2}}(j)$ is the probability that node D receives the SNR of value $(\frac{j-1}{\mathcal{N}}\Gamma_{th}\leq \gamma_{overall1}+\gamma_{SD}<\frac{j}{\mathcal{N}}\Gamma_{th})$ when system has transitioned the TCs state from $\{S,R1\}$ to $\{S,R1,R2\}_1$, while $p_{d_{R1->R2}}(j)$ is the probability that node D receives the SNR of value $(\frac{j-1}{\mathcal{N}}\Gamma_{th}\leq \gamma_{overall1}+\gamma_{R1D}<\frac{j}{\mathcal{N}}\Gamma_{th})$ when system has transitioned the TCs state from $\{S,R1\}$ to $\{S,R1,R2\}_2$. Obviously, the expressions of $p_{d_{S->R2}}(j)$ and $p_{d_{R1->R2}}(j)$ are as follows:
\begin{equation}
	\begin{split}
		p_{d_{S->R2}}(j)=\mathrm{conv}(p_{\gamma_{overall1}},p_{SD})\\
	\end{split}
\end{equation}
\begin{equation}
	\begin{split}
		p_{d_{R1->R2}}(j)=\mathrm{conv}(p_{\gamma_{overall1}},p_{R1D})\\
	\end{split}
\end{equation}
where 
\begin{equation}
	\begin{split}
		p_{SD}(j)=\int_{\frac{j-1}{\mathcal{N}}\Gamma_{th}}^{\frac{j}{\mathcal{N}}\Gamma_{th}}f_{SD}(x)dx\\
	\end{split}
\end{equation}
\begin{equation}
	\begin{split}
		p_{R1D}(j)=\int_{\frac{j-1}{\mathcal{N}}\Gamma_{th}}^{\frac{j}{\mathcal{N}}\Gamma_{th}}f_{R1D}(x)dx.\\
	\end{split}\label{T32}
\end{equation}
Similarly,
\begin{equation}
	\begin{split}
		p_{R2D}(j)=\int_{\frac{j-1}{\mathcal{N}}\Gamma_{th}}^{\frac{j}{\mathcal{N}}\Gamma_{th}}f_{R2D}(x)dx.\\
	\end{split}\label{R2D}
\end{equation}

With these parameter, the expressions of $c$, $\mathfrak{d}$, $\mathfrak{f}$, $g$, $m$, $n$ and $o$ can be derived.
\begin{equation}
	\begin{split}
		\mathfrak{c}=\sum_{j=1}^{\mathcal{N}}\mathrm{conv}(p_{\gamma_{overall1}},p_{SD})
	\end{split}\label{c}
\end{equation}
\begin{equation}
	\begin{split}
		\mathfrak{d}=\sum_{j=1}^{\mathcal{N}}\mathrm{conv}(p_{\gamma_{overall1}},p_{R1D})
	\end{split}
\end{equation}
\begin{equation}
	\begin{split}
		\mathfrak{m}=\sum_{j=1}^{\mathcal{N}}\mathrm{conv}(p_{\gamma_{overall2}},p_{SD})
	\end{split}
\end{equation}
\begin{equation}
	\begin{split}
		\mathfrak{n}=\sum_{j=1}^{\mathcal{N}}\mathrm{conv}(p_{\gamma_{overall2}},p_{R2D})
	\end{split}
\end{equation}
\begin{equation}
	\begin{split}
		\mathfrak{f}=\sum_{j=1}^{\mathcal{N}}\mathrm{conv}(p_{\gamma_{overall3}},p_{R2D})
	\end{split}
\end{equation}
\begin{equation}
	\begin{split}
		\mathfrak{g}=\sum_{j=1}^{\mathcal{N}}\mathrm{conv}(p_{\gamma_{overall3}},p_{R1D})
	\end{split}
\end{equation}
\begin{equation}
	\begin{split}
		\mathfrak{o}=\sum_{j=1}^{\mathcal{N}}\mathrm{conv}(p_{\gamma_{overall3}},p_{SD})
	\end{split}\label{o}
\end{equation}
The detailed process for obtaining probability $p_s $, $p_{SR1}$, $p_{SR2}$, $p_1$, $p_2$, $p_3$, $p_{\gamma_{overall1}}(j)$, $p_{\gamma_{overall2}}(j)$, $p_{\gamma_{overall3}}(j)$, '$\mathfrak{c}$', '$\mathfrak{d}$', '$\mathfrak{f}$', '$\mathfrak{g}$', '$\mathfrak{m}$', '$\mathfrak{n}$' and '$\mathfrak{o}$' is shown in Alg. \ref{Alg.1}. See Appendix A. 
\par As the preparations are ready, in sections A and B, the PDF of energy in the PEB of R1 and R2 are obtained, respectively.
\subsection{Limiting Distribution of R2 Energy}
By using the DCSMC and donating the energy level in the infinite-size energy buffer in the $i$-th signaling interval as $B_2(i)$, the PEB's buffer of R2 update equations are given by
%%Similar to solution of the limiting distribution of R1 energy, by using the DCSMC and donating the energy level in the infinite-size energy buffer in the $i$-th signaling interval as $B_2(i)$, the PEB`s buffer of R2 update equations are given by
\begin{equation}\label{eq24}
	\begin{split}
		& B_2(i+1) = B_2(i)+X(i), \quad \mathbb{P}_{21}\\
		& B_2(i+1) = B_2(i)-M_{R2}+X(i), \quad \mathbb{P}_{22}
	\end{split}
\end{equation}
where $\mathbb{P}_{21}$ and $\mathbb{P}_{22}$ are the applicable conditions of the update equations, respectively.
They are given by
\begin{equation}\label{eq25}
	\begin{split}
		\mathbb{P}_{21}:\, & \big( TCs = \{S\} \big) \cup \big(TCs = \{S,R1\} \big) \\
		& \cup \Big[\big(TCs = \{S,R2\} \big) \cap \big(B_2(i) < M_{R2}\big)\Big]\\
		& \cup \Big[\big(TCs = \{S,R2\} \big) \cap \big(B_2(i) \ge M_{R2}\big) \cap \mathbb{M} \cap \mathbb{N} \Big]\\
		& \cup \Big[\big(TCs = \{S,R2\} \big) \cap \big(B_2(i) \ge M_{R2}\big) \cap \overline{\mathbb{M}}  \Big]\\
		&\cup \Big[\big(TCs = \{S,R1,R2\}\big) \cap \big(B_2(i) \ge M_{R2} \cap \mathbb{O} \cap \mathbb{F} \big)\Big]\\
		&\cup \Big[\big(TCs = \{S,R1,R2\}\big) \cap \big(B_2(i) \ge M_{R2} \cap \overline{\mathbb{O}}  \big)\Big]\\
		&\cup \Big[\big(TCs = \{S,R1,R2\}\big) \cap \big(B_2(i) < M_{R2} \big)\Big]
		,
	\end{split}
\end{equation}
\begin{equation}\label{eq26}
	\begin{split}
		\mathbb{P}_{22}:\, &\Big[\big(TCs = \{S,R2\} \big) \cap \big(B_2(i) \ge M_{R2}\big)\cap \mathbb{C} \cap \overline{\mathbb{N}} \Big]\\
		&\cup \Big[\big(TCs = \{S,R1,R2\}\big) \cap \big(B_2(i) \ge M_{R2} \cap \mathbb{O} \cap \overline{\mathbb{F}} \big)\Big].
	\end{split}
\end{equation}
Let $\psi_{R2}$ represents the energy buffer convergence parameter for the energy update process $B_1(i)$ of Eq. (\ref{eq24}), with the help of the Eq. (\ref{c}) to  Eq. (\ref{o})  , which can be expressed as follows
%%Similarly, let $\psi_{R2}$ represents the energy buffer convergence parameter for the energy update process $B_2(i)$ of Eq. (\ref{eq24}), which can be expressed as follows
\begin{equation}\label{eq30}
	\begin{split}
		\psi_{R2}  & = \lambda_2 M_{R2} \Big[p_{SR2}*\mathfrak{m}(1-\mathfrak{n})+(p_1+p_2+p_3)*\mathfrak{o}(1-\mathfrak{f})\Big]\\
		&= \lambda_2 M_{R2}b_2,
	\end{split}
\end{equation}
	where 
\begin{equation}\label{b_2}
	\begin{split}
		b_2  & =p_{SR2}*\mathfrak{m}(1-\mathfrak{n})+(p_1+p_2+p_3)*o(1-\mathfrak{f}).
	\end{split}
\end{equation}
The limiting PDF when $\psi_{R2} > 1$ is obtained in Theorem 1, and the scenario when $\psi_{R2} \le 1$ is discussed in Theorem 2.
\begin{theorem}\label{theo3}
	If $\psi_{R2} > 1$, the energy update process $B_2(i)$ in Eq. (\ref{eq24}) will have a stationary distribution. Furthermore, the limiting PDF of the energy buffer state at R2 can be expressed by
	\begin{equation}\label{eq31}
		g_2(x) = \begin{cases}
			\dfrac{1}{M_{R2}}\left(1-e^{Q_2 x}\right), & 0 \leq x < M_{R2} \\
			\dfrac{1}{M_{R2}}\dfrac{-Q_2 e^{Q_2 x}}{\left(b_2\lambda_2+Q_2\right)}, & x \ge M_{R2}.
		\end{cases}
	\end{equation}
Where
	\begin{equation}\label{Q_2}
		\begin{split}
			Q_2 = \frac{-W\left(-b_2\lambda_2 M_{R2} e^{-b_2\lambda_2 M_{R2}}\right)}{M_{R2}}-b_2\lambda_2 < 0,
		\end{split}
	\end{equation}
	satisfying $b_2 \lambda_2 e^{Q_2 M_{R2}} = b_2 \lambda_2 + Q_2$.
\end{theorem}

\begin{proof}
	The proof is given in Appendix B.
\end{proof}
\begin{theorem}\label{theo4}
	When $\psi_{R2} \leq 1$, there is no stationary PDF of the energy buffer state of R1. In addition, after a finite number of time slots, the R1 energy buffer will keep $B_2(i) > M_{R2}$.
	%%Similar to Theorem \ref{theo2}, when $\psi_2 \leq 1$, there is no stationary PDF of the energy buffer state of R2. And the energy in R2 energy buffer keep more than $M_{R2}$.
\end{theorem}
\begin{proof}
	The proof is given in Appendix C.
\end{proof}
\subsection{Limiting Distribution of R1 Energy}
	Similar to the solution of the limiting distribution of R2 energy, by using the DCSMC and donating the energy level in the infinite-size energy buffer in the $i$-th signaling interval as $B_1(i)$, the PEB's buffer of R1 update equations are given by
\begin{equation}\label{eq10}
	\begin{split}
		& B_1(i+1) = B_1(i)+X(i), \quad \mathbb{P}_{11}\\
		& B_1(i+1) = B_1(i)-M_{R1}+X(i), \quad \mathbb{P}_{12}
	\end{split}
\end{equation}
where $\mathbb{P}_{11}$ and $\mathbb{P}_{12}$ are the update conditions needed for the equations in Eq. (\ref{eq10}) according to OR protocol. Correspondingly, they can be expressed as follows
\begin{equation}\label{eq12}
	\begin{split}
		\mathbb{P}_{12}:\, &\Big[\big(TCs = \{S,R1\} \big) \cap \big(B_1(i) \ge M_{R1}\big)\cap \mathbb{C} \cap \overline{\mathbb{D}} \Big]\\
		&\cup \Big[\big(TCs = \{S,R1\} \big) \cap \big(\big(B_1(i) \ge M_{R1}\big), \mathbb{C} , \mathbb{D},\\
		&\qquad  \gamma_{SR2}(i)<\Gamma_{th},\gamma_{R1R2}(i)\ge\Gamma_{th}\big) \Big]\\
		&\cup \Big[\big(TCs = \{S,R1,R2\}\big) \cap \big(B_2(i) \ge M_{R2} , \mathbb{O} , \mathbb{F} ,\\
		&\big(B_1(i) \ge M_{R1}\big),\overline{\mathbb{G}}\big)\Big]\\
		&\cup \Big[\big(TCs = \{S,R1,R2\}\big) \cap \big(B_2(i) < M_{R2} , \mathbb{O},\\
		&\big(B_1(i) \ge M_{R1}\big),\overline{\mathbb{G}}\big)\Big],
	\end{split}
\end{equation}
\begin{equation}\label{eq11}
	\begin{split}
		\mathbb{P}_{11}:\, & \big(  TCs=\{S\} \big) \cup \big( TCs=\{S,R2\} \big)\\
		& \cup \Big[\big(TCs = \{S,R1\} \big) \cap \big(B_1(i) < M_{R1}\big)\Big]\\
		& \cup \Big[\big(TCs = \{S,R1\} \big) \cap \big(B_1(i) \ge M_{R1},\mathbb{C},\mathbb{D},\\
		&\qquad \gamma_{SR2}(i)\ge\Gamma_{th},\big)\Big]\\
		& \cup \Big[\big(TCs = \{S,R1\} \big) \cap \big(B_1(i) \ge M_{R1},\mathbb{C},\mathbb{D},\\
		&\qquad \gamma_{SR2}(i)<\Gamma_{th},\gamma_{R1R2}(i)<\Gamma_{th}\big)\Big]\\
		& \cup \Big[\big(P1 \big) \cap \big(B_2(i) \ge M_{R2},\mathbb{O}, \overline{\mathbb{F}} \big)\Big]\\
		& \cup \Big[\big(TCs = \{S,R1,R2\} \big) \cap \big(B_2(i) \ge M_{R2},\mathbb{O}, \mathbb{F},\\
		&B_1(i) \ge M_{R1},\mathbb{G} \big)\Big]\\
		& \cup \Big[\big(TCs = \{S,R1,R2\} \big) \cap \big(B_2(i) \ge M_{R2},\mathbb{O}, \mathbb{F},\\
		&B_1(i) < M_{R1} \big)\Big]\\
		& \cup \Big[\big(TCs = \{S,R1,R2\} \big) \cap \big(B_2(i) < M_{R2},\mathbb{O},\\
		&B_1(i) \ge M_{R1},\mathbb{G} \big)\Big]\\
		& \cup \Big[\big(TCs = \{S,R1,R2\} \big) \cap \big(B_2(i) < M_{R2},\mathbb{O},\\
		&B_1(i) < M_{R1} \big)\Big].
		%& \cup \big(\Omega = R2 \big)
	\end{split}
\end{equation}

Let $\psi_{R1}$ represent the energy buffer convergence parameter for the energy update process $B_1(i)$ of Eq. (\ref{eq10}), which can be expressed as follows
\begin{equation}\label{eq22}
	\begin{split}
		\psi_{R1}  & = \lambda_1 M_{R1}\Big[p_{SR1}*\mathfrak{c}((1-\mathfrak{d})+\mathfrak{d}(1-e^{-W_{SR2}\Gamma_{th}})e^{-W_{R1R2}\Gamma_{th}}\\
		&+PU2((p_1+p_2+p_3)*\mathfrak{o}*\mathfrak{f}(1-\mathfrak{g})\\
		&+(1-PU2)*((p_1+p_2+p_3)*\mathfrak{o}(1-\mathfrak{g}))\Big]\\
		&= \lambda_1 M_{R1}b_1 ,
	\end{split}
\end{equation}
where 
\begin{equation}\label{b_1}
	\begin{split}
		b_1  & = p_{SR1}*\mathfrak{c}((1-\mathfrak{d})+\mathfrak{d}(1-e^{-W_{SR2}\Gamma_{th}})e^{-W_{R1R2}\Gamma_{th}}\\
		&+PU2((p_1+p_2+p_3)*\mathfrak{o}*\mathfrak{f}(1-\mathfrak{g})\\
		&+(1-PU2)*((p_1+p_2+p_3)*\mathfrak{o}(1-\mathfrak{g})).
	\end{split}
\end{equation}
The limiting PDF when $\psi_{R1} > 1$ is obtained in Theorem 3, and the scenario when $\psi_{R1} \le 1$ is discussed in Theorem 4.
\begin{theorem}\label{theo1}
	If $\psi_1 > 1$, the energy update process $B_1(i)$ in Eq. (\ref{eq10}) will have a stationary distribution. Furthermore, the limiting PDF of the energy buffer state at R1 can be expressed by
	\begin{equation}\label{eq23}
		g_1(x) = \begin{cases}
			\dfrac{1}{M_{R1}}\left(1-e^{Q_1 x}\right), & 0 \leq x < M_{R1} \\
			\dfrac{1}{M_{R1}}\dfrac{-Q_1 e^{Q_1 x}}{\left(b_1\lambda_1+Q_1\right)}, & x \ge M_{R1}.
		\end{cases}
	\end{equation}

And
\begin{equation}\label{Q_1}
	\begin{split}
		Q_1 = \frac{-W\left(-b_1\lambda_1 M_R e^{-b_1\lambda_1 M_{R1}}\right)}{M_{R1}}-b_1\lambda_1 < 0,
	\end{split}
\end{equation}
 satisfying $b_1 \lambda_1 e^{Q_1 M_{R1}} = b_1 \lambda_1 + Q_1$.

\end{theorem}

\begin{proof}
	The proof is given in Appendix D.
\end{proof}

\begin{theorem}\label{theo2}
	Similar to Theorem \ref{theo4}, when $\psi_{R1} \leq 1$, there is no stationary PDF of the energy buffer state of R1. And the energy in the R1 energy buffer keeps more than $M_{R1}$.
\end{theorem}

\section{Analysis of Outage Probability and Throughput}
This section focuses on the analysis of outage probability,  timeslot cost, and throughput of our system, when both $B_1(i)$ and $B_2(i)$ have limiting PDFs.

According to OR protocol, when all the transmitter candidates fail to transmit signals to the D node, the communication network system is defined as an outage. In an outage, transmitter candidates can transmit the signal to other nodes. Obviously, the system outage probability (OP) can be expressed as
\begin{equation}\label{SEC IV1}
	\begin{split}
		OP = rp_{S}+rp_{SR1}+rp_{SR2}+rp_{SR1R2},
	\end{split}
\end{equation}
where, $rp_{S}$, $rp_{SR1}$, $rp_{SR2}$ and $rp_{SR1R2}$ are the outage probability of transmitter candidates set \{S\}, \{S,R1\}, \{SR2\}, \{S,R1,R2\}, respectively. The outage probability of the node means the transmitter candidates as the broadcast node fail to transmit signals to the D node. Specifically, they can be expressed by

\begin{equation}\label{SEC IV2}
	\begin{split}
		rp_S & = p_S \mathrm{Pr}\{\gamma_{SD}(i) <\Gamma_{th}\}\\
		& = p_S \left(1-e^{-W_{SD}\Gamma_{th}}\right),
	\end{split}
\end{equation}
\begin{equation}\label{SEC IV4}
	\begin{split}
		rp_{SR1} & = p_{SR1} [\mathrm{Pr}\{\mathbb{C},\mathbb{D},B_1(i) \ge M_{R1}\}\\
		&+\mathrm{Pr}\{\mathbb{C},B_1(i) < M_{R1}\} ]\\
		&=p_{SR1}*\mathfrak{c}(\mathfrak{d}*PU1+1-PU1),
	\end{split}
\end{equation}
\begin{equation}\label{SEC IV3}
	\begin{split}
		rp_{SR2} & =  p_{SR2} [\mathrm{Pr}\{\mathbb{M},\mathbb{N},B_2(i) \ge M_{R2}\}
		\\&+\mathrm{Pr}\{\mathbb{M},B_2(i) < M_{R2}\} ]\\
		&=p_{SR2}*\mathfrak{m}(\mathfrak{n}*PU2+1-PU2),
	\end{split}
\end{equation}

\begin{equation}\label{SEC IV}
	\begin{split}
		&rp_{SR1R2} \\
		& =  (p_{1}+p_2+p_3) \Big[\mathrm{Pr}\{\mathbb{O},\mathbb{F},B_2(i) \ge M_{R2},\mathbb{G},B_1(i) \ge M_{R1}\}\\
		&+\mathrm{Pr}\{\mathbb{O},\mathbb{F},B_2(i) \ge M_{R2},B_1(i) < M_{R1}\}\\
		&+\mathrm{Pr}\{\mathbb{O},B_2(i) < M_{R2},\mathbb{G},B_1(i) \ge M_{R1}\}\\
		&+\mathrm{Pr}\{\mathbb{O},B_2(i) < M_{R2},B_1(i) < M_{R1}\}\Big]\\
		&=(p_{1}+p_2+p_3)(\mathfrak{o}(\mathfrak{f}*PU2+1-PU2)\\
		&\qquad\times(\mathfrak{g}*PU1+1-PU1)).
	\end{split}
\end{equation}

With the Eq. (\ref{eq31}), we get
\begin{equation}\label{SEC IV5}
	\begin{split}
		\mathrm{Pr}\{B_2(i) < M_{R2} \} & =  \int_{x=0}^{M_{R2}} \dfrac{\left(1-e^{Q_2 x}\right)}{M_{R2}}\, dx \\
		& = 1-\dfrac{\left(e^{Q_2 M_{R2}}-1\right)}{Q_2 M_{R2}}.
	\end{split}
\end{equation}
Due to $b_2 \lambda_2 e^{Q_2 M_{R2}} = b_2 \lambda_2 + Q_2$, we get $e^{Q_2 M_{R2}} = 1+\frac{Q_2}{b_2\lambda_2}$, and thus
\begin{equation}\label{SEC IV6}
	\begin{split}
		\mathrm{Pr}\{B_2(i) < M_{R2}\} & = 1-\dfrac{\left(1+\frac{Q_2}{b_2\lambda_2}-1\right)}{Q_2 M_{R2}}\\
		& = 1-\dfrac{1}{b_2\lambda_2 M_{R2}}.
	\end{split}
\end{equation}
Accordingly, we get
\begin{equation}\label{SEC IV7}
	\begin{split}
		\mathrm{Pr}\{B_2(i) \ge M_{R2}\} & = 1-\mathrm{Pr}\{B_2(i) < M_{R2}\}\\
		& = \dfrac{1}{b_2\lambda_2 M_{R2}}=PU2.
	\end{split}
\end{equation}
Similarly, with the Eq. (\ref{eq31}) and $b_1 \lambda_1 e^{Q_1 M_{R1}} = b_1 \lambda_1 + Q_1$, we can get
\begin{equation}\label{SEC IV17}
	\begin{split}
		\mathrm{Pr}\{B_1(i) < M_{R1}\}& = 1-\dfrac{1}{b_1\lambda_1 M_{R1}},
	\end{split}
\end{equation}
and
\begin{equation}\label{SEC IV18}
	\begin{split}
		\mathrm{Pr}\{B_1(i) \ge M_{R1}\}& =\dfrac{1}{b_1\lambda_1 M_{R1}}=PU1.
	\end{split}
\end{equation}
With the help of Eq. (\ref{c}) to  Eq. (\ref{o}) and Eq. (\ref{b_1}) as well as Eq. (\ref{b_2}) and Eq. (\ref{SEC IV1}) to  Eq. (\ref{SEC IV}) can be rewritten as follows
\begin{equation}\label{SEC IV19}
	\begin{split}
		OP & = p_S(1-e^{-W_{SD}\Gamma_{th}})\\
		&+p_{SR1}\Big(\mathfrak{c}*(\mathfrak{d}*PU1+1-PU1)\Big)\\
		&+p_{SR2}\Big(\mathfrak{c}*(\mathfrak{n}*PU2+1-PU2)\Big)\\
		&+(p_{1}+p_2+p_3)\Big(\mathfrak{o}(f*PU2+1-PU2)\\
		&\times(\mathfrak{g}*PU1+1-PU1)\Big).
	\end{split}
\end{equation}
Correspondingly, the throughput of the system $\tau$ can be defined as 
\begin{equation}\label{SEC IV20}
	\begin{split}
		\tau  & = \left(1 - OP\right)R_0\\
		& = R_0\Bigg[1-\Bigg(p_S(1-e^{-W_{SD}\Gamma_{th}})\\
		&+p_{SR1}\Big(\mathfrak{c}*(\mathfrak{d}*PU1+1-PU1)\Big)\\
		&+p_{SR2}\Big(\mathfrak{c}*(\mathfrak{n}*PU2+1-PU2)\Big)\\
		&+(p_{1}+p_2+p_3)\Big(\mathfrak{o}(\mathfrak{f}*PU2+1-PU2)\\
		&\times(\mathfrak{g}*PU1+1-PU1)\Big)\Bigg)\Bigg] .
	\end{split}
\end{equation}

Furthermore, the timeslot cost for each data $T_c$ is a performance index worthy of attention. And it can be calculated as follows:
\begin{equation}
	\begin{split}
		T_c  & = \lim_{N \to \infty}\sum_{i=1}^N OP^{i-1}(1-OP)i\\
		&=\lim_{N \to \infty} (1-OP) \frac{1}{1-OP}[1+\frac{(OP)(1-(OP)^{N-1})}{1-OP}\\
		&-N(OP)^{N}]\\
		&=\frac{1}{1-OP}.
	\end{split}
\end{equation}

\section{performance results}

In this section, simulations are presented to verify the rationality of the derived theoretical expressions. In addition, the performance comparison between applied MRC and non-applied MRC is performed. For all simulations, the following system parameters are taken into account unless otherwise specified. Suppose S, R1, R2, and D are all located on a two-dimensional plane, and their position coordinates are (0,0), (30,20), (60,-20), and (100,0), respectively. Path-loss exponent $\alpha = -3$. The transmitting power of S node $P_{Ts} = 12 $dBm. The channel noise variance $N_0=-50$dBm. With the measure of Alg. \ref{Alg.1}, the corresponding transmitter candidate set probabilities, $p_{\gamma_{overall1}}$, $p_{\gamma_{overall2}}$, $p_{\gamma_{overall3}}$, '$\mathfrak{c}$', '$\mathfrak{d}$', '$\mathfrak{f}$', '$\mathfrak{g}$', '$\mathfrak{m}$', '$\mathfrak{n}$' and '$\mathfrak{o}$' are obtained. Furthermore, in all figures, point markers represent simulation values, and solid red lines indicate the STM-based theoretical values.

In Fig. \ref{fig5}, three subfigures depict the PDF of the value of SNR received by D in different TCs, namely $p_{\gamma_{overall1}}$, $p_{\gamma_{overall2}}$, $p_{\gamma_{overall3}}$. The red solid line represents the theoretic value calculated with Alg. \ref{Alg.1}. It can be observed that the simulation values, represented by the blue circle markers, lay on the red solid line, which shows the rationality of the Alg. \ref{Alg.1}.
\begin{figure*}[h!]

		\begin{minipage}{0.33\linewidth}
			\vspace{3pt}
			%这个图片路径替换成你的图片路径即可使用
			\centerline{\includegraphics[width=\textwidth]{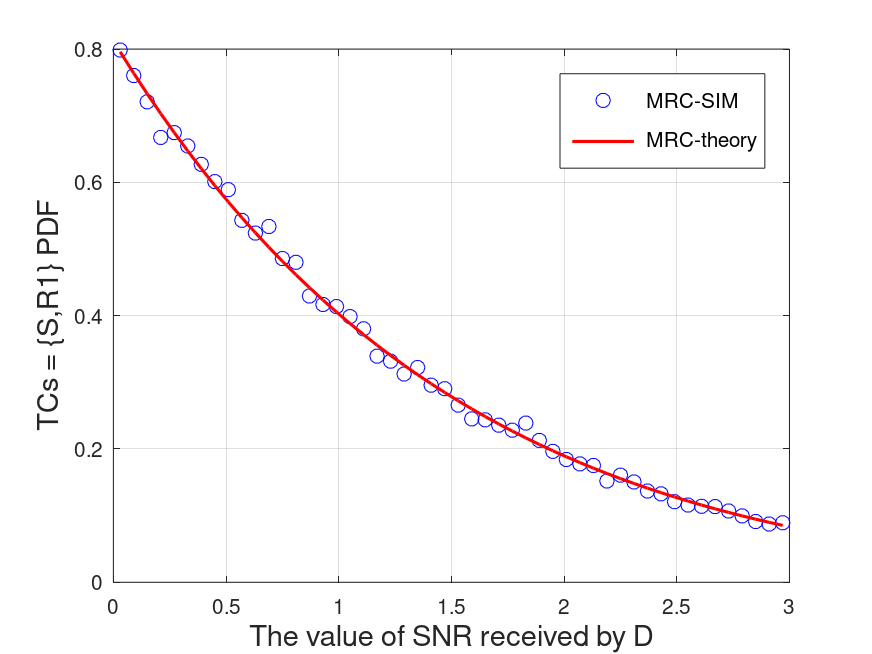}}
			% 加入对这列的图片说明
			%\centerline{Image}
		\end{minipage}
		\begin{minipage}{0.33\linewidth}
			\vspace{3pt}
			\centerline{\includegraphics[width=\textwidth]{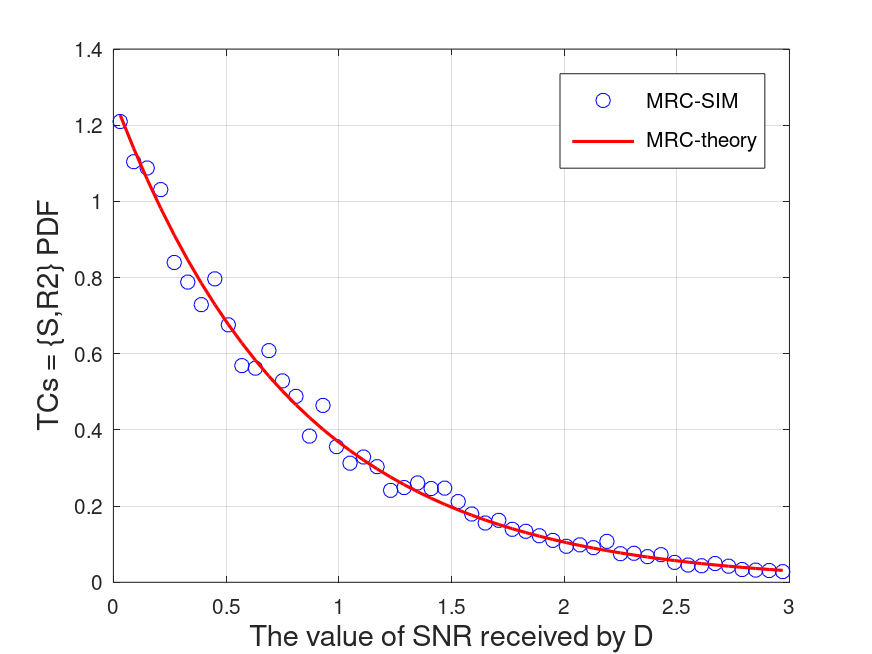}}
			
			%\centerline{Image}
		\end{minipage}
		\begin{minipage}{0.33\linewidth}
	\vspace{3pt}
	\centerline{\includegraphics[width=\textwidth]{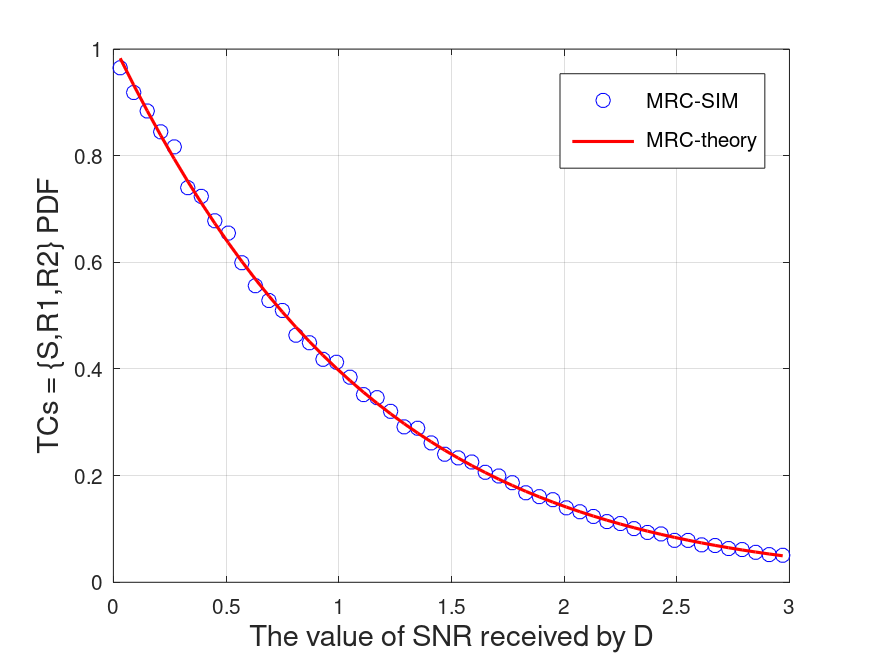}}
	%\centerline{Image}
\end{minipage}
	\caption{PDF of the value of SNR received by D in different TCs with  parameters $\lambda_1 = -12$ dB, $\lambda_2 = -15$ dB, $R_0=2$ bit/s/Hz, $M_1 = 10$ mJoules, $M_2= 8$ mJoules.}
	\label{fig5}
\end{figure*}

Fig. \ref{fig1} and Fig. \ref{fig8} depict the limiting PDF of energy stored in the R1 buffer and R2 buffer, obtained by simulation and theoretical analysis. It can be clearly seen from Fig. \ref{fig1} that the solid line and dashed line almost exactly coincide, and they are almost in the middle of the simulation values. Furthermore, the same case can be seen in Fig. \ref{fig8}, which effectively proves the accuracy of the derived theoretical expressions in Eq. (\ref{eq23}) and Eq. (\ref{eq30}).

\begin{figure}[h!]
	\centerline{\includegraphics[width=3.5in]{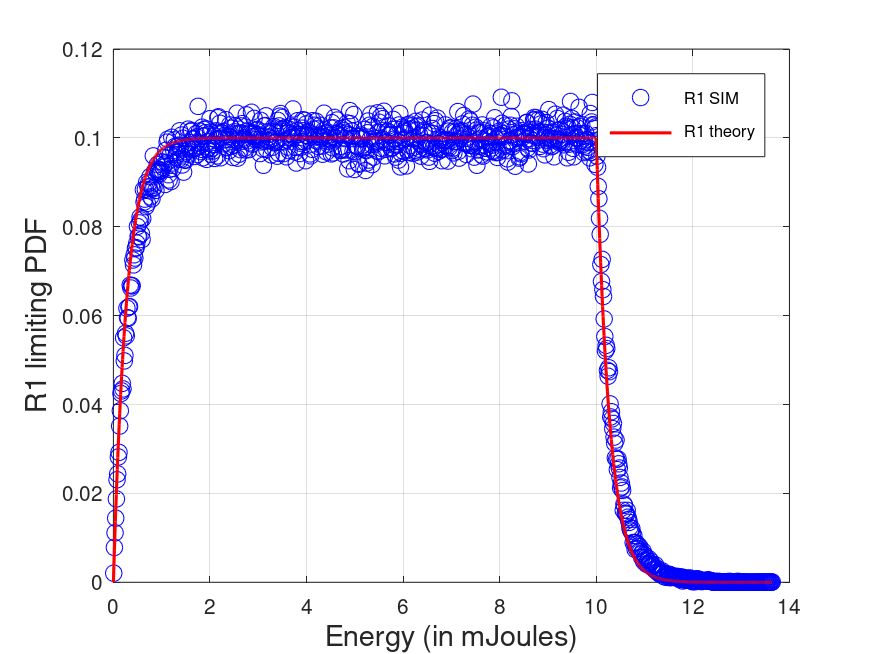}}
	\begin{comment}
		\begin{minipage}{0.49\linewidth}
			\vspace{3pt}
			%这个图片路径替换成你的图片路径即可使用
			\centerline{\includegraphics[width=\textwidth]{R1 pdf 41}}
			% 加入对这列的图片说明
			%\centerline{Image}
		\end{minipage}
		\begin{minipage}{0.49\linewidth}
			\vspace{3pt}
			\centerline{\includegraphics[width=\textwidth]{R2 pdf 41}}
			
			%\centerline{Image}
		\end{minipage}
	\end{comment}
	\caption{Limiting distribution of energy with infinite-size PEB of relay 1 and parameters $\lambda_1 = -15$ dB, $\lambda_2 = -12$ dB, $R_0=2$ bit/s/Hz, $M_1 = 10$ mJoules, $M_2= 8$ mJoules.}
	\label{fig1}
\end{figure}

\begin{figure}
	\centerline{\includegraphics[width=3.5in]{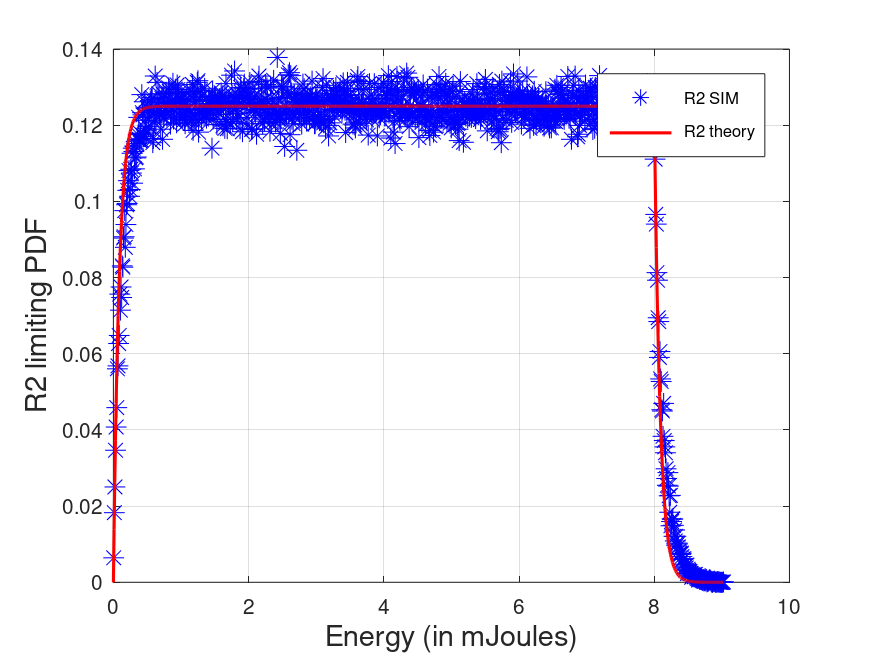}}
	\begin{comment}
		\begin{minipage}{0.49\linewidth}
			\vspace{3pt}
			%这个图片路径替换成你的图片路径即可使用
			\centerline{\includegraphics[width=\textwidth]{R1 pdf 41}}
			% 加入对这列的图片说明
			%\centerline{Image}
		\end{minipage}
		\begin{minipage}{0.49\linewidth}
			\vspace{3pt}
			\centerline{\includegraphics[width=\textwidth]{R2 pdf 41}}
			
			%\centerline{Image}
		\end{minipage}
	\end{comment}
	\caption{Limiting distribution of energy with finite-size PEB of relay 2 and parameters $\lambda_1 = -15$ dB, $\lambda_2 = -12$ dB, $R_0=2$ bit/s/Hz, $M_1 = 10$ mJoules, $M_2= 8$ mJoules..}
	\label{fig8}
\end{figure}

\begin{figure}[h]
	\centerline{\includegraphics[width=3.5in]{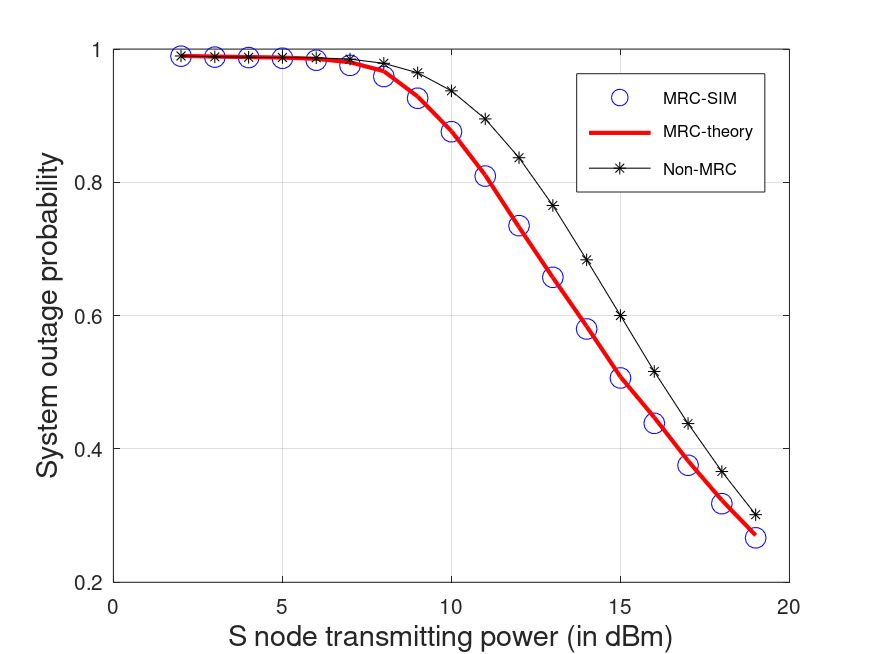}}
	\begin{comment}
	\begin{minipage}{0.32\linewidth}
		\vspace{3pt}
		%这个图片路径替换成你的图片路径即可使用
		\centerline{\includegraphics[width=\textwidth]{R1 Pout lamda}}
		% 加入对这列的图片说明
	\end{minipage}
	\begin{minipage}{0.32\linewidth}
		\vspace{3pt}
		\centerline{\includegraphics[width=\textwidth]{R2 Pout lamda}}
	\end{minipage}
	\begin{minipage}{0.32\linewidth}
		\vspace{3pt}
		\centerline{\includegraphics[width=\textwidth]{sym Pout lamda}}
	\end{minipage}
	\end{comment}
	\caption{Outage probability of system vs. $P_{Ts}$, with parameters $\lambda_1 = -11$ dB, $\lambda_2 = -12$ dB, $R_0=2$ bit/s/Hz, $M_1 = 12$ mJoules, $M_2= 10$ mJoules.}
	\label{fig2}
\end{figure}
\par Fig. \ref{fig2} shows the system outage probability of the system applied MRC obviously less than Non-MRC system. In additon, Fig. \ref{fig3} depicts the throughput of system and our system also outperform to Non-MRC system. This is because the MRC algorithm can combine the multipath signal with the appropriate coefficient to achieve diversity gain. Furthermore, it can be obvserved that the outage and throughput change slightly, when $P_{Ts}$ is from 2 dBm to 8 dBm, and change considerably when $P_{Ts}$ is from 9 dBm to 19 dBm. The reason for this phenomenon is that, when $P_{Ts}$ is small, it is difficult for S node transmitting data packet reliably to any other node, in other words, in the whole system, only S node is activated in most time. With the increase in $P_{Ts}$, S node is able to transmit the data packet to relay nodes. With the help of relay nodes, the performance of system improves substantially.

\begin{figure}[h]
		\centerline{\includegraphics[width=3.5in]{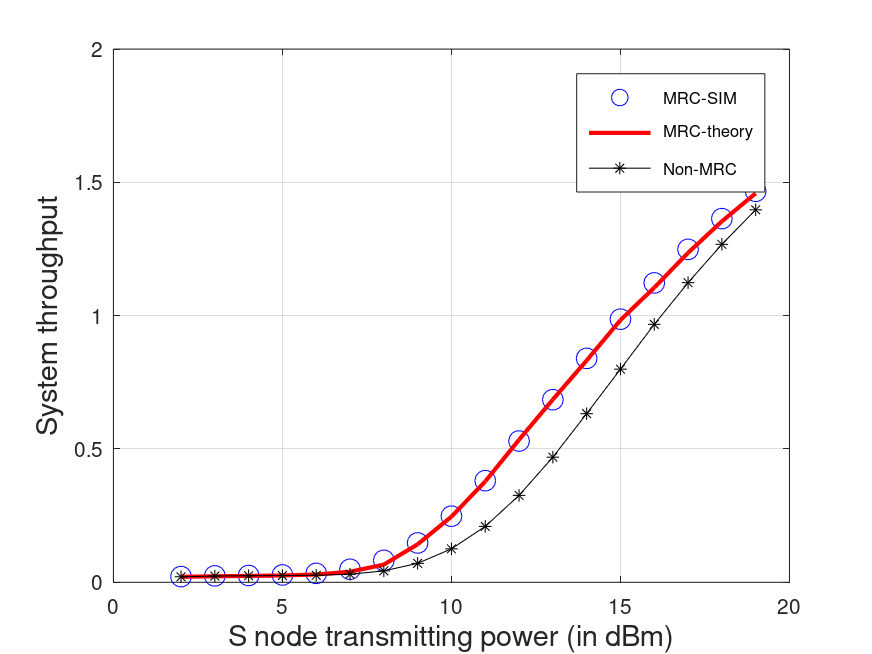}}
			\begin{comment}
	\begin{minipage}{0.32\linewidth}
		\vspace{3pt}
		%这个图片路径替换成你的图片路径即可使用
		\centerline{\includegraphics[width=\textwidth]{R1 Pout}}
		% 加入对这列的图片说明
	\end{minipage}
	\begin{minipage}{0.32\linewidth}
		\vspace{3pt}
		\centerline{\includegraphics[width=\textwidth]{R2 Pout}}
	\end{minipage}
	\begin{minipage}{0.32\linewidth}
		\vspace{3pt}
		\centerline{\includegraphics[width=\textwidth]{sym pout}}
	\end{minipage}
		\end{comment}
	\caption{Throughput of system vs. $P_{Ts}$, with with parameters $\lambda_1 = -11$ dB, $\lambda_2 = -12$ dB, $R_0=2$ bit/s/Hz, $M_1 = 12$ mJoules, $M_2= 10$ mJoules.}
	\label{fig3}
	
\end{figure}

\begin{figure}[h]
			\centerline{\includegraphics[width=3.5in]{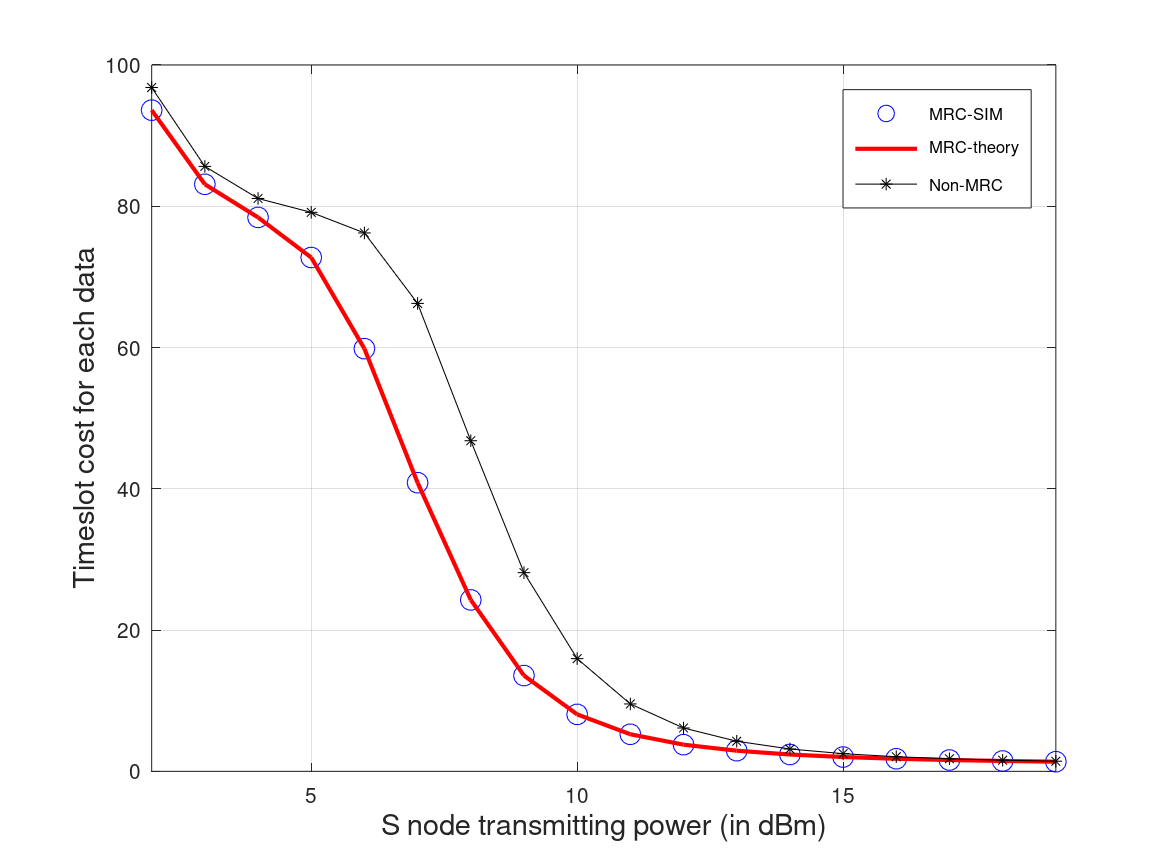}}
				\begin{comment}
	\begin{minipage}{0.32\linewidth}
		\vspace{3pt}
		%这个图片路径替换成你的图片路径即可使用
		\centerline{\includegraphics[width=\textwidth]{R1 Pout 41}}
		% 加入对这列的图片说明
	\end{minipage}
	\begin{minipage}{0.32\linewidth}
		\vspace{3pt}
		\centerline{\includegraphics[width=\textwidth]{R2 Pout 41}}
	\end{minipage}
	\begin{minipage}{0.32\linewidth}
		\vspace{3pt}
		\centerline{\includegraphics[width=\textwidth]{sym pout41}}
	\end{minipage}
				\end{comment}
	\caption{Timeslot cost for each data vs. $P_{Ts}$, with parameters $\lambda_1 = -11$ dB, $\lambda_2 = -12$ dB, $R_0=2$ bit/s/Hz, $M_1 = 12$ mJoules, $M_2= 10$ mJoules.}
	\label{fig4}
	
\end{figure}
In Fig. \ref{fig4}, the system timeslot cost for each data are demonstrated. These simulations can clarify that the MRC system outperforms the Non-MRC system. Furthermore, it can be seen that timeslot cost of MRC system is close to that of Non-MRC system when $P_{Ts}$ is from 2 dBm to 4 dBm and from 14 dBm to 19dBm. This is because, when $P_{Ts}$ is too small, the SNR of the signals combined is also too small to reach the SNR threshold and transmit the data packet. As a reasult, the effect of MRC is not obvious. And when $P_{Ts}$ is too large, even the Non-MRC system can transmit the data packet to D easily so that the effect of MRC is not significant. Thus, if $P_{Ts}$ is from 5 dBm to 13 dBm, the Non-MRC system is hard to utilize the SNR of single signal to reach SNR threshold, however, our MRC system can exploit the SNR of multiple signals combined to finish the data packet transmitting. Hence, in these cases, MRC significantly improves system performance.

\section{Conclusion}
This paper combines OR aided cooperative communication network with EH nodes and MRC. By applying the discrete-time continuous-state space Markov chain model, the theoretical algorithm-based expressions for limiting the distributions of stored energy in infinite-size buffers are derived. Furthermore, by using both the limiting distributions of energy buffers and data-broadcasted probability at nodes, the theoretical numerical expressions for outage probability and throughput of the network are obtained. In the future, combining the generalized selection combining (GSC) with OR under the assumption of an EH cooperative communication network is worth being studied.

\appendices
\section{Algorithm for Probability $p_s $, $p_{SR1}$, $p_{SR2}$, $p_1$, $p_2$, $p_3$, $p_{\gamma_{overall1}}$, $p_{\gamma_{overall2}}$, $p_{\gamma_{overall3}}$, '$\mathfrak{c}$', '$\mathfrak{d}$', '$\mathfrak{f}$', '$\mathfrak{g}$', '$\mathfrak{m}$', '$\mathfrak{n}$' and '$\mathfrak{o}$' Based on STM}
Specifically, in line 9, 26, 44 and 58, the judgment condition $(p(i)\leq 0\,\vert\,\textbf{T}(i)<0)$ indicates that there are non-positive elements in $\textbf{p}(i)$ or $\textbf{T}(i)$, which is not desirable, so the update process needs to be terminated. Moreover, in line 12, 29, 47 and 61, $(||\textbf{p}(i)-\textbf{p}(i+1)||_{2}<10^{-7})$ indicates that the update error is small enough and the update process has converged. Therefore, the update process can be terminated.
\begin{breakablealgorithm}
	\caption{probability $p_s $, $p_{SR1}$, $p_{SR2}$, $p_1$, $p_2$, $p_3$, $p_{\gamma_{overall1}}$, $p_{\gamma_{overall2}}$, $p_{\gamma_{overall3}}$, '$\mathfrak{c}$', '$\mathfrak{d}$', '$\mathfrak{f}$', '$\mathfrak{g}$', '$\mathfrak{m}$', '$\mathfrak{n}$' and '$\mathfrak{o}$' based on STM}
	\label{Alg.1}
	\begin{algorithmic}[1]
		\Require $W_{SD}$, $W_{SR1}$, $W_{SR2}$, $W_{R1D}$, $W_{R1R2}$, $W_{R2D}$, $\Gamma_{th}$, $\lambda_1$, $\lambda_2$, $M_{R1}$, $M_{R2}$, $\mathcal{N}$.
		\State Initialize $i_4=0$, $k_4=0$, $j_1$ and $j_2 \in \{1,2,3,4,5,6\}$, $\textbf{p}(0)=[p_s, p_{SR1}, p_{SR2}, p_1, p_2, p_3]=[\frac{1}{6}, \frac{1}{6}, \frac{1}{6}, \frac{1}{6}, \frac{1}{6}, \frac{1}{6}]$, $PU1=PU2=\frac{1}{2}$.
		\While{1}
		\State Calculate $\textbf{T1}(i_4)$ and $\textbf{T2}(i_4)$ according to equations from Eq. (\ref{T1}) to Eq. (\ref{T2}),
		\State$i_1=0$, $k_1=0$
		\State$\textbf{p}_{\gamma_{overall1}}(0)=[p_{\gamma_{overall1}}(1),...,p_{\gamma_{overall1}}(\mathcal{N})]$
		\State $p_{\gamma_{overall1}}(j)=\frac{1}{\mathcal{N}}$, $1\leq j\leq \mathcal{N}$ and $j\in Z$
		\While{1}
		\State $\textbf{p}_{\gamma_{overall1}}(i_1+1) =\textbf{p}_{\gamma_{overall1}}(i_1)\textbf{T1}(i_4)$,
		\If {$(\textbf{p}_{\gamma_{overall1}}(i)\leq 0\,\vert\ \vert\,\textbf{T1}(i_4)<0)$}
		\State $k_1=i_1-1$,
		\State break,
		\ElsIf {$(||\textbf{p}_{\gamma_{overall1}}(i_1)-\textbf{p}_{\gamma_{overall1}}(i_1+1)||_{2}<10^{-7})$}
		\State $k_1=i_1$,
		\State break,
		\Else
		\State $i_1=i_1+1$,
		\EndIf
		\EndWhile
		\State$p_{\gamma_{overall1}}=\textbf{p}_{\gamma_{overall1}}(k_1)$
		\State
		\State$i_2=0$, $k_2=0$
		\State$\textbf{p}_{\gamma_{overall2}}(0)=[p_{\gamma_{overall2}}(1),...,p_{\gamma_{overall2}}(\mathcal{N})]$
		\State $p_{\gamma_{overall2}}(j)=\frac{1}{\mathcal{N}}$, $1\leq j\leq \mathcal{N}$ and $j\in Z$
		\While{1}
		\State $\textbf{p}_{\gamma_{overall2}}(i_2+1) =\textbf{p}_{\gamma_{overall2}}(i_2)\textbf{T2}(i)$,
		\If {$(\textbf{p}_{\gamma_{overall2}}(i)\leq 0\,\vert\ \vert\,\textbf{T2}(i_4)<0)$}
		\State $k_2=i_2-1$,
		\State break,
		\ElsIf {$(||\textbf{p}_{\gamma_{overall2}}(i_2)-\textbf{p}_{\gamma_{overall2}}(i_2+1)||_{2}<10^{-7})$}
		\State $k_2=i_2$,
		\State break,
		\Else
		\State $i_2=i_2+1$,
		\EndIf
		\EndWhile
		\State$p_{\gamma_{overall2}}=\textbf{p}_{\gamma_{overall2}}(k_2)$
		\State
		\State Calculate $\textbf{T3}(i_4)$, $p_{d_{S->R2}}$ and $p_{d_{R1->R2}}$ according to equations from Eq. (\ref{T31}) to Eq. (\ref{T32}),
		\State$i_3=0$, $k_3=0$
		\State$\textbf{p}_{\gamma_{overall3}}(0)=[p_{\gamma_{overall3}}(1),...,p_{\gamma_{overall3}}(\mathcal{N})]$
		\State $p_{\gamma_{overall3}}(j)=\frac{1}{\mathcal{N}}$, $1\leq j\leq \mathcal{N}$ and $j\in Z$
		\While{1}
		\State $\textbf{p}_{\gamma_{overall3}}(i_3+1) =\textbf{p}_{\gamma_{overall3}}(i_3)\textbf{T3}(i_4)$,
		\If {$(\textbf{p}_{\gamma_{overall3}}(i)\leq 0\,\vert\ \vert\,\textbf{T3}(i_4)<0)$}
		\State $k_3=i_3-1$,
		\State break,
		\ElsIf {$(||\textbf{p}_{\gamma_{overall3}}(i_3)-\textbf{p}_{\gamma_{overall3}}(i_3+1)||_{2}<10^{-7})$}
		\State $k_3=i_3$,
		\State break,
		\Else
		\State $i_3=i_3+1$,
		\EndIf
		\EndWhile
		\State$p_{\gamma_{overall3}}=\textbf{p}_{\gamma_{overall3}}(k_3)$
		\State Calculate '$\mathfrak{c}$', '$\mathfrak{d}$', '$\mathfrak{f}$', '$\mathfrak{g}$', '$\mathfrak{m}$', '$\mathfrak{n}$' and '$\mathfrak{o}$' and equations from Eq. (\ref{R2D}) to Eq. (\ref{o}),
		\State Calculate $\textbf{T}(i)$ according to $\textbf{p}(i)$ and equations from Eq. (\ref{P_s}) to Eq. (\ref{p_{(SR1R2)_3-(SR1R2)_3}}),
		\State $\textbf{p}(i_4+1) =\textbf{p}(i_4)\textbf{T}(i_4)$,
		\If {$(\textbf{p}_{j_1}(i)\leq 0\,\vert\ \vert\,\textbf{T}_{j_1,j_2}(i_4)<0)$}
		\State $k_4=i_4-1$,
		\State break,
		\ElsIf {$(||\textbf{p}(i_4)-\textbf{p}(i_4+1)||_{2}<10^{-7})$}
		\State $k_=i_4$,
		\State break,
		\Else
		\State $i_4=i_4+1$,
		\EndIf
		\EndWhile
		\Ensure Probability Distribution $\textbf{p}(k_4)$, $p_{\gamma_{overall1}}$, $p_{\gamma_{overall2}}$, $p_{\gamma_{overall3}}$, '$\mathfrak{c}$', '$\mathfrak{d}$', '$\mathfrak{f}$', '$\mathfrak{g}$', '$\mathfrak{m}$', '$\mathfrak{n}$' and '$\mathfrak{o}$'.
	\end{algorithmic}
\end{breakablealgorithm}

\section{Proof of Theorem 1}
According to the total probability theorem, $\mathbb{P}_{21}$ and $\mathbb{P}_{22}$,  the CDF of $B_2(i+1)$ in storage process in Eq. (\ref{eq24}) may be evaluated as follows

\begin{equation}\label{APPENDIX C 1}
	\begin{split}
		& \mathrm{Pr}\{B_2(i+1) \leq x\} \\
		& = \mathrm{Pr}\{B_2(i)+X(i) \leq x, TCs = \{S\}\} \\
		&  + \mathrm{Pr}\{B_2(i)+X(i) \leq x, TCs = \{S,R1\} \}\\
		&  + \mathrm{Pr}\{B_2(i)+X(i) \leq x, TCs = \{S,R2\}, B_2(i) < M_{R2}\} \\
		&  + \mathrm{Pr}\{B_2(i)+X(i) \leq x, TCs = \{S,R2\}, B_2(i) \ge M_{R2},\mathbb{M},\mathbb{N} \} \\
		&  + \mathrm{Pr}\{B_2(i)+X(i) \leq x, TCs = \{S,R2\}, B_2(i) \ge M_{R2},\overline{\mathbb{C}}\} \\
		&  + \mathrm{Pr}\{B_2(i)+X(i) \leq x, TCs = \{S,R1,R2\},\\
		&\qquad \qquad B_2(i) \ge M_{R2},\mathbb{O},\mathbb{F} \} \\
		&  + \mathrm{Pr}\{B_2(i)+X(i) \leq x, TCs = \{S,R1,R2\},\\
		&\qquad \qquad B_2(i) \ge M_{R2},\overline{\mathbb{O}} \} \\
		&  + \mathrm{Pr}\{B_2(i)+X(i) \leq x, TCs = \{S,R1,R2\}, B_2(i) < M_{R2} \} \\
		&  + \mathrm{Pr}\{B_2(i)-M_{R2} +X(i) \leq x,TCs = \{S,R2\},\\
		&\qquad \qquad  B_2(i) \ge M_{R2},\mathbb{M},\overline{\mathbb{N}}\} \\
		&  + \mathrm{Pr}\{B_2(i)-M_{R2} +X(i) \leq x, TCs = \{S,R1,R2\},\\
		& \qquad \qquad B_2(i) \ge M_{R2},{\mathbb{O}},\overline{\mathbb{F}}\},
	\end{split}
\end{equation}
When $i \to \infty$, the equation $\mathrm{Pr}\{B_2(i+1) \leq x\} = G^{i+1}_2(x) = G^{i}_2(x) = G_2(x)$ holds for the steady state of energy buffer at RNs. Thus, Eq. (\ref{APPENDIX C 1}) can be written as follows
\begin{equation}\label{APPENDIX C 2}
	\begin{split}
		G_2(x) & = \left(p_S+p_{SR1} \right) \int_{\mu_2=0}^{x} F_X(x-\mu_2)g_2(\mu_2)\, d\mu_2 \\
		&  + \big(p_{SR2}*(\mathfrak{m}*\mathfrak{n}+1-\mathfrak{m})\\
		&+(p_1+p_2+p_3)(\mathfrak{o}*\mathfrak{f}+1-\mathfrak{o})\big)\\
		& \times\int_{\mu_2=M_{R2}}^{x} F_X(x-\mu_2)g_2(\mu_2)d\mu_2\\
		&  + (p_{SR2}+p_{SR1R2}) \int_{\mu_2=0}^{min(x,M_{R2})} F_X(x-\mu_2)g_2(\mu_2)\, d\mu_2 \\
		&  + \big(p_{SR2}*\mathfrak{m}*(1-\mathfrak{n})+(p_1+p_2+p_3)*\mathfrak{o}*(1-\mathfrak{f})\big)\\
		&  \times \int_{\mu_2=M_{R2}}^{x+M_{R2}} F_X(x+M_{R2}-\mu_2) g_2(\mu_2)\, d\mu_2,
	\end{split}
\end{equation}
where $g_2(x)$ is the derivative of $G_2(x)$. With simplifying Eq. (\ref{APPENDIX C 2}), we arrive at
\begin{equation}\label{APPENDIX C 3}
	G_2(x) = \begin{cases}
		G_{21}(x) , & 0 \leq x < M_{R2} \\
		G_{22}(x), & x \ge M_{R2}
	\end{cases}
\end{equation}
where,
\begin{equation}\label{APPENDIX C 4}
	\begin{split}
		G_{21}(x) & = b_2\int_{\mu_2=M_{R2}}^{x+M_{R2}} F_X(x+M_{R2}-\mu_2)g_1(\mu_2)\, d\mu_2 \\
		& \quad + \int_{\mu_2=0}^{x} F_X(x-\mu_2)g_1(\mu_2)\, d\mu_2 , 0 \leq x < M_{R2}
	\end{split}
\end{equation}
\begin{equation}\label{APPENDIX C 5}
	\begin{split}
		G_{22}(x) & = \int_{\mu_2=0}^{M_{R2}} F_X(x-\mu_2) g_2(\mu_2)\, d\mu_2 \\
		& \quad +a_2 \int_{\mu_2=M_{R2}}^{x} F_X(x-\mu_2)g_2(\mu_2)\, d\mu_2 \\
		& \quad + b_2\int_{\mu_2=M_{R2}}^{x+M_{R2}} F_X(x+M_{R2}-\mu_2) g_2(\mu_2)\, d\mu_2,\\
		& \qquad \qquad \qquad \qquad \qquad \qquad \qquad \quad x \ge M_{R2}
	\end{split}
\end{equation}
where,
\begin{equation}\label{APPENDIX C9}
	\begin{split}
		a_2&=p_S+p_{SR1}+p_{SR2}*\mathfrak{m}*\mathfrak{n}\\
		&+(p_1+p_2+p_3)\big(\mathfrak{o}*\mathfrak{f}+(1-\mathfrak{o})\big)+P_{SR2}*(1-\mathfrak{m}),
	\end{split}
\end{equation}
\begin{equation}\label{APPENDIX C10}
	\begin{split}
		b_2& = p_{SR2}*\mathfrak{m}(1-\mathfrak{n})+(p_1+p_2+p_3)*\mathfrak{o}(1-\mathfrak{f}).
	\end{split}
\end{equation}
According to Eq. (\ref{APPENDIX C 3}), the PDF $g_2(x)$ may be defined as
\begin{equation}\label{APPENDIX C 6}
	g_2(x) = \begin{cases}
		g_{21}(x), & 0 \leq x < M_{R2} \\
		g_{22}(x), & x \ge M_{R2}.
	\end{cases}
\end{equation}
Substituting Eq. (\ref{APPENDIX C 6}) into Eq. (\ref{APPENDIX C 5}), the derivatives about $x$ on both sides of Eq. (\ref{APPENDIX C 5}) can be obtained as follows
\begin{equation}\label{APPENDIX C 7}
	\begin{split}
		g_{22}(x) & = \int_{\mu_1=0}^{M_{R2}} f_X(x-\mu_2)g_{21}(\mu_2)\, d\mu_2 \\
		& \quad + a_2 \int_{\mu_1=M_{R2}}^{x} f_X(x-\mu_2)g_{22}(\mu_2)\, d\mu_2 + b_2 \\
		& \quad \times \int_{\mu_1=M_{R2}}^{x+M_{R2}} f_X(x+M_{R2}-\mu_2)g_{22}(\mu_2)\, d\mu_2, \\
		&\qquad \qquad \qquad\qquad \qquad \qquad \qquad x \ge M_{R2}
	\end{split}
\end{equation}

Obviously, $a_2+b_2 = 1$. In addtion,with \cite{article3}, it is declared that, when $M_{R2}>\mathbb{E}[X(i)]=1/\lambda_2$, the energy update process in Eq. (\ref{eq24}) possesses a unique stationary distribution, namely, $g_2(x)$ has unique solution. Furthermore, form \cite{article3}, $g_{22}(x)$ can be assumed to have an exponential form solution, which can be expressed by $g_{22}(x) = k_2 e^{Q_2 x}$. Substituting $g_{22}(x) = k_2 e^{Q_2 x}$ and $f_X(x) = \lambda_2 e^{-\lambda_2 x}$ into Eq. (\ref{APPENDIX C 7}), we obtain

\begin{equation}\label{APPENDIX C 10}
	\begin{split}
		k_2 e^{Q_2 x} & = \int_{\mu_2=0}^{M_{R2}} \lambda_2 e^{-\lambda_2 (x-\mu_2)} g_{21}(\mu_2)\, d\mu_2 \\
		& \quad + a_2 \int_{\mu_1=M_{R2}}^{x} \lambda_2 e^{-\lambda_2 (x-\mu_2)} k_2 e^{Q_2 \mu_2}\, d\mu_2 + b_2 \\
		& \quad \times \int_{\mu_1=M_{R2}}^{x+M_{R2}} \lambda_2 e^{-\lambda_2 (x+M_{R2}-\mu_2)} k_2 e^{Q_2 \mu_2}\, d\mu_2,\\
		&\qquad \qquad \qquad\qquad \qquad \qquad \qquad x \ge M_{R2}
	\end{split}
\end{equation}
Simplifying Eq. (\ref{APPENDIX C 10}), we obtain
\begin{equation}\label{APPENDIX C 11}
	\begin{split}
		k_2 e^{Q_2 x} & = \lambda_2 e^{-\lambda_2 x}\int_{\mu_2=0}^{M_{R2}} e^{\lambda_2 \mu_2} g_{21}(\mu_2)\, d\mu_2 \\
		& \quad - \dfrac{k_2 \lambda_2 e^{\left(Q_2 M_{R2}-\lambda_2 x \right)} \left(a_2 e^{\lambda_2 M_{R2}}+b_2 \right)}{\lambda_2 + Q_2} \\
		& \quad + \dfrac{\lambda_2 \left( a_2 + b_2 e^{Q_2 M_{R2}} \right)k_2 e^{Q_2 x}}{\lambda_2 + Q_2}, x \ge M_{R2}
	\end{split}
\end{equation}
\newpage
	\begin{strip}
For Eq. (\ref{APPENDIX C 11}) to hold, the following conditions need to be satisfied

\begin{subequations}
	\begin{numcases}{}
		\dfrac{\lambda_2 \left( a_2 + b_2 e^{Q_2 M_{R2}} \right)}{\lambda_2 + Q_2} = 1, \label{APPENDIX C 12a}\\
		\dfrac{k_2 e^{Q_2 M_{R2}} \left(a_2 e^{\lambda_2 M_{R2}}+b_2 \right)}{\lambda_2 + Q_2} = \int_{\mu_2=0}^{M_{R2}} e^{\lambda_2 \mu_2} g_{21}(\mu_2)\, d\mu_2. \label{APPENDIX C 12b}
	\end{numcases}
\end{subequations}
	\end{strip}
The desirable solution $Q_2$ of Eq. (\ref{APPENDIX C 12a}) for the finite distribution of $g_{22}(x)$ may be obtained by simplifying Eq. (\ref{APPENDIX C 12a}) as follows
\begin{equation}\label{APPENDIX C 13}
	b_2 \lambda_2 e^{Q_2 M_{R2}} = \lambda_2 - a_2 \lambda_2 + Q_2= b_2 \lambda_2 + Q_2,
\end{equation}

Although, it can be easily found from Eq. (\ref{APPENDIX C 12a}) that $Q_{2_0} = 0$ is one of the solutions of $Q_2$ in Eq. (\ref{APPENDIX C 12a}), $Q_{2_0}$ does not satisfy the condition that $g_{22}(x)$ is a limiting distribution. And the other solution $Q_{2_1}$ of $Q_2$ in Eq. (\ref{APPENDIX C 12a}) can be obtained by simplifying Eq. (\ref{APPENDIX C 12a}) as
\begin{equation}\label{APPENDIX C 14}
	Q_2 = \frac{-W\left(-b_2\lambda_2 M_{R2} e^{-b_2\lambda_2 M_{R2}}\right)}{M_{R2}}-b_2\lambda_2, \quad b_2\lambda_2 M_{R2}>1.
\end{equation}
where, with $b_2\lambda_2 M_{R2} >1$, $W\left(-b_2\lambda_2 M_{R2} e^{-b_2\lambda_2 M_{R2}}\right)> -b_2\lambda_2 M_{R2}$ so that $Q_2<0$, ensuring the finite distribution of $g_{22}(x)$.

Correspondingly, when $0\leq x <M_{R2}$, substitute Eq. (\ref{APPENDIX C 6}) into Eq. (\ref{APPENDIX C 4}) and differentiate both sides of Eq. (\ref{APPENDIX C 4}) with respect to $x$. we have
\begin{equation}\label{APPENDIX C 16}
	\begin{split}
		g_{21}(x) & = b_2 \int_{\mu_2=M_{R2}}^{x+M_{R2}} f_X(x+M_{R2}-\mu_2)g_{22}(\mu_2)\, d\mu_2\\
		& \quad + \int_{\mu_2=0}^{x} f_X(x-\mu_2)g_{21}(\mu_2)\, d\mu_2, 0 \leq x <M_{R2}
	\end{split}
\end{equation}
Substituting $g_{22}(x)=k_2 e^{Q_2 x}$ and $f_X(x)=\lambda_2 e^{-\lambda_2 x}$ into Eq. (\ref{APPENDIX C 16}), we get
\begin{equation}\label{APPENDIX C 17}
	\begin{split}
		g_{21}(x) & = \lambda_2 \int_{\mu_2=0}^{x} e^{-\lambda_2 \left(x-\mu_2 \right)}g_{21}(\mu_2)\, d\mu_2 \\
		& \quad + \dfrac{b_2 k_2 \lambda_2 e^{Q_2 M_{R2}}}{\lambda_2 + Q_2} \left(e^{Q_2 x}-e^{-\lambda_2 x} \right), 0 \leq x <M_{R2}
	\end{split}
\end{equation}
Let $r(x)=\frac{b_2 k_2 \lambda_2 e^{Q_2 M_{R2}}}{\lambda_2 + Q_2} \left(e^{Q_2 x}-e^{-\lambda_2 x} \right)$, and the integral equation in Eq. (\ref{APPENDIX C 17}) can be rewritten as follows
\begin{equation}\label{APPENDIX C19}
	\begin{split}
		g_{21}(x) = \lambda_2 \int_{\mu_2=0}^{x} e^{-\lambda_2 \left(x-\mu_2 \right)}g_{21}(\mu_2)\, d\mu_2 + r(x) , 0 \leq x <M_{R2}
	\end{split}
\end{equation}
Clearly, Eq. (\ref{APPENDIX C19}) is a Volterra integral equation of the second kind, whose solution is given by \cite{article3}, \cite{article4} and \cite[eq. 2.2.1]{30}
\begin{equation}\label{APPENDIX B20}
	\begin{split}
		g_{21}(x) = r(x) + \lambda_2 \int_{t=0}^{x} r(t)\, dt,
	\end{split}
\end{equation}
Substituting $r(x)$ into Eq. (\ref{APPENDIX B20}), we obtain
\begin{equation}\label{APPENDIX C 20}
	\begin{split}
		g_{21}(x) & = \dfrac{b_2 k_2 \lambda_2 e^{Q_2 M_{R2}}}{\lambda_2 + Q_2} \left(e^{Q_2 x}-e^{-\lambda_2 x} \right) \\
		& \quad + \lambda_2 \int_{t=0}^{x} \dfrac{b_2 k_2 \lambda_2 e^{Q_2 M_{R2}}}{\lambda_2 + Q_2} \left(e^{Q_2 t}-e^{-\lambda_2 t} \right)\, dt \\
		& = \dfrac{b_2 k_2 \lambda_2 e^{Q_2 M_{R2}}\left(e^{Q_2 x}-1 \right)}{Q_2}, 0 \leq x <M_{R2}
	\end{split}
\end{equation}
Because of the unit area condition on $g_2(x)$, we have
\begin{equation}\label{APPENDIX C 21}
	\int_{x=0}^{\infty} g_2(x)\, dx = \int_{x=0}^{M_{R2}} g_{21}(x)\, dx +\int_{x=M_{R2}}^{\infty} g_{22}(x)\, dx = 1,
\end{equation}
Substituting $g_{21}(x) = \frac{b_2 k_2 \lambda_2 e^{Q_2 M_{R2}}\left(e^{Q_2 x}-1 \right)}{Q_2}$ and $g_{22}(x) = k_2 e^{Q_2 x}$ into Eq. (\ref{APPENDIX C 21}), we arrive at
\begin{equation}\label{APPENDIX C 24}
	k_2 = \dfrac{-Q_2}{M_{R2} \left(b_2 \lambda_2 + Q_2\right)},
\end{equation}
Furthermore, according to Eq. (\ref{APPENDIX C 24}), we have
\begin{equation}\label{APPENDIX C 25}
	g_{21}(x) =  \dfrac{1-e^{Q_2 x}}{M_{R2}}.
\end{equation}

Substituting Eq. (\ref{APPENDIX C 25}) into the right side of Eq. (\ref{APPENDIX C 12b}), we obtain
\begin{equation}\label{APPENDIX B277}
	\begin{split}
		\int_{\mu_2=0}^{M_{R2}} e^{\lambda_2 \mu_2} g_{21}(\mu_2)\, d\mu_2 & = \int_{\mu_2=0}^{M_{R2}} \dfrac{\left(1-e^{Q_2 x}\right)e^{\lambda_2 \mu_2}}{M_{R2}}\, d\mu_2 \\
		& = \dfrac{1-e^{\left(\lambda_2+Q_2\right)M_{R2}}}{\left(\lambda_2+Q_2\right)M_{R2}} - \dfrac{1-e^{\lambda_2 M_{R2}}}{\lambda_2 M_{R2}},
	\end{split}
\end{equation}
The equation in Eq. (\ref{APPENDIX C 12a}) leads us to conclude $\lambda_2 M_{R2} = \frac{\left(\lambda_2+Q_2\right)M_{R2}}{b_2 e^{Q_2 M_{R2}} + a_2}$. Substituting this conclusion in Eq. (\ref{APPENDIX B277}), we have
\begin{equation}\label{APPENDIX B287}
	\begin{split}
		\int_{\mu_2=0}^{M_{R2}} e^{\lambda_2 \mu_2} g_{21}(\mu_2)\, d\mu_2 = \dfrac{\left(1-e^{Q_2 M_{R2}}\right)\left(b_2+a_2 e^{\lambda_2 M_{R2}}\right)}{\left(\lambda_2+Q_2\right)M_{R2}},
	\end{split}
\end{equation}
Similarly, the conclude $1-e^{Q_2 M_{R2}} = \frac{-Q_2}{b_2\lambda_2}$ may be obtained from Eq. (\ref{APPENDIX C 13}). Substituting this conclusion in Eq. (\ref{APPENDIX B287}), we arrive at
\begin{equation}\label{APPENDIX C 28}
	\begin{split}
		\int_{\mu_2=0}^{M_{R2}} e^{\lambda_2 \mu_2} g_{21}(\mu_2)\, d\mu_2 & =\dfrac{\left(1-e^{Q_2 M_{R2}}\right)\left(b_2+a_2 e^{\lambda_2 M_{R2}}\right)}{\left(\lambda_2+Q_2\right)M_{R2}} \\
		& = \dfrac{-Q_2 \left(b_2+a_2 e^{\lambda_2 M_{R2}}\right)}{\left(\lambda_2+Q_2\right)M_{R2} b_2\lambda_2} \\
		& = \dfrac{-Q_2 e^{Q_2 M_{R2}} \left(b_2+a_2 e^{\lambda_2 M_{R2}}\right)}{M_{R2} \left(b_2\lambda_2 + Q_2\right) \left(\lambda_2+Q_2\right)} \\
		& = \dfrac{k_2 e^{Q_2 M_{R2}} \left(b_2+a_2 e^{\lambda_2 M_{R2}}\right)}{\lambda_2+Q_2}.
	\end{split}
\end{equation}
Now, the validation of Eq. (\ref{APPENDIX C 12b}) in  Eq. (\ref{APPENDIX C 28}) indicates that the unique solution $g_{21}(x)$ in Eq. (\ref{APPENDIX C 25}) for Eq. (\ref{APPENDIX C 16}) and the unique solution $g_{22}(x) = k_2 e^{Q_2 x}$ for Eq. (\ref{APPENDIX C 7}) are obtained.

\section{Proof of Theorem 2}
According to the energy storage process $B_2(i)$ with the infinite-size energy buffer in Eq. (\ref{eq10}), the variable $O_R(i)$ is defined by
\begin{equation}\label{APPENDIX A1}
	O_R(i) = \begin{cases}
		1,\quad \mathbb{P}_{I1} \\
		0,\quad \mathbb{P}_{II1}
	\end{cases}
\end{equation}
where $\mathbb{P}_{I1}$ and $\mathbb{P}_{II1}$ are the conditions required for the equations in Eq. (\ref{APPENDIX A1}) to hold. Specifically, they can be expressed as follows
\begin{equation}\label{APPENDIX A2}
	\begin{split}
		\mathbb{P}_{I1}:\,&\Big[\big(TCs = \{S,R2\} \big) \cap \mathbb{C} \cap \overline{\mathbb{N}} \Big]\\
		&\cup \Big[\big(TCs = \{S,R1,R2\}\big)  \cap \mathbb{O} \cap \overline{\mathbb{F}} \Big],
	\end{split}
\end{equation}
\begin{equation}\label{APPENDIX A3}
	\begin{split}
		\mathbb{P}_{II1}:\,& \big( TCs =\{S\} \big) \cup \big(TCs = \{S,R1\} \big) \\
		& \cup \Big[\big(TCs = \{S,R2\} \big)  \cap \mathbb{C} \cap \mathbb{N} \Big]\\
		& \cup \Big[\big(TCs = \{S,R2\} \big) \cap \overline{\mathbb{C}}  \Big]\\
		&\cup \Big[\big(TCs = \{S,R1,R2\}\big)  \cap \mathbb{O} \cap \mathbb{F}\Big]\\
		&\cup \Big[\big(TCs = \{S,R1,R2\}\big)  \cap \overline{\mathbb{O}}\Big].
	\end{split}
\end{equation}
Moreover, the variable $\theta(i)$ is defined as follows
\begin{equation}\label{APPENDIX A4}
	\theta(i) = \begin{cases}
		1,\quad B_2(i) \ge M_{R2} \\
		0,\quad B_2(i) < M_{R2}
	\end{cases}
\end{equation}
Implying Eq. (\ref{APPENDIX A1}), Eq. (\ref{APPENDIX A2}),Eq. (\ref{APPENDIX A3}) and Eq. (\ref{APPENDIX A4}), the energy storage process in Eq. (\ref{eq10}) can be re-expressed as follows
\begin{equation}\label{APPENDIX A5}
	B_2(i+1) - B_2(i) = X(i) - M_{R2}\theta(i)O_R(i).
\end{equation}
According to the law of large numbers, the average energy harvesting rate can be obtained as follows
\begin{equation}\label{APPENDIX A6}
	\mathbb{E}[X(i)] =  \lim_{N \to \infty}\frac{1}{N}\sum_{i=1}^N X(i),
\end{equation}
Similarly, the average energy consumption rate can be given by
\begin{equation}\label{APPENDIX A7}
	\begin{split}
		\mathbb{E}[M_{R2}\theta(i)O_R(i)] & = \lim_{N \to \infty}\frac{1}{N}\sum_{i=1}^N M_{R2}\theta(i)O_R(i) \\
		& = \lim_{N \to \infty}\frac{1}{N}\sum_{B_2(i)\ge M_{R2}} M_{R2} O_R(i) \\
		& \leq M_{R2}\Big\{\lim_{N \to \infty}\frac{1}{N}\sum_{i=1}^N O_R(i)\Big\},
	\end{split}
\end{equation}
And,
\begin{equation}\label{APPENDIX A8}
	\begin{split}
		\mathbb{E}[O_R(i)] & =  \lim_{N \to \infty}\frac{1}{N}\sum_{i=1}^N O_R(i) \\
		& = 1 \times \mathrm{Pr}\{O_R(i) = 1\} + 0 \times \mathrm{Pr}\{O_R(i) = 0\} \\
		& = \mathrm{Pr}\{B_2(i)-M_{R2} +X(i) \leq x,TCs = \{S,R2\},\\
		&\qquad \qquad  B_2(i) \ge M_{R2},\mathbb{M},\overline{\mathbb{N}}\} \\
		&  + \mathrm{Pr}\{B_2(i)-M_{R2} +X(i) \leq x, TCs = \{S,R1,R2\},\\
		&\qquad \qquad B_2(i) \ge M_{R2}, \mathbb{O},\overline{\mathbb{F}}\} \\
		&=b_2,
	\end{split}
\end{equation}
From Eq. (\ref{APPENDIX A7}) and Eq. (\ref{APPENDIX A8}), we obtain
\begin{equation}\label{APPENDIX A9}
	\begin{split}
		\mathbb{E}[M_{R2}\theta(i)O_R(i)] & \leq  b_1.
	\end{split}
\end{equation}
If $\psi_{R2} \leq 1$, from Eq. (\ref{eq30}), we get
\begin{equation}\label{APPENDIX A10}
	\begin{split}
		\mathbb{E}[X(i)] & \ge  M_{R2} b_2,
	\end{split}
\end{equation}
Therefore, if $\psi_1 \leq 1$, from Eq. (\ref{APPENDIX A6}), Eq. (\ref{APPENDIX A7}), Eq. (\ref{APPENDIX A9}) and Eq. (\ref{APPENDIX A10}), we obtain
\begin{equation}\label{APPENDIX A11}
	\lim_{N \to \infty}\frac{1}{N}\sum_{i=1}^N X(i) \ge \lim_{N \to \infty}\frac{1}{N}\sum_{i=1}^N M_{R2}\theta(i)O_R(i),
\end{equation}
According to Eq. (\ref{APPENDIX A5}) and Eq. (\ref{APPENDIX A11}), we have
\begin{equation}\label{APPENDIX A12}
	\lim_{N \to \infty}\frac{1}{N}\sum_{i=1}^N B_2(i+1) \ge \lim_{N \to \infty}\frac{1}{N}\sum_{i=1}^N B_2(i),
\end{equation}
Clearly, if the inequality condition holds in Eq. (\ref{APPENDIX A12}), the comparison of between $B_1(i+1)$ and $B_1(i)$ shows that the energy accumulates in the buffer over time slot, i.e., $\begin{matrix} \lim_{i \to \infty}\mathbb{E}[B_1(i)] = \infty\end{matrix}$. Therefore, the stationary distribution of $B_1(i)$ does not exist, and after a finite number of time slots, $B_1(i) > M_{R1}$ almost always hold. In addition, if the equality condition holds in Eq. (\ref{APPENDIX A11}), according to Eq. (\ref{APPENDIX A11}) and Eq. (\ref{APPENDIX A7}), we get
\begin{equation}\label{APPENDIX A13}
	\lim_{N \to \infty}\frac{1}{N}\sum_{i=1}^N X(i) = \lim_{N \to \infty}\frac{1}{N}\sum_{B_2(i)\ge M_{R2}} M_{R2} O_R(i).
\end{equation}
Eq. (\ref{APPENDIX A13}) indicates that in the energy buffer with DCSMC model, the average energy harvesting rate equals the average energy consumption rate, which is unstable \cite{29}. Therefore, the buffer may almost always provide $M_{R2}$ amount of energy.

\section{Proof of Theorem 3}
On the basis of the total probability theorem, $\mathbb{P}_{11}$ and $\mathbb{P}_{12}$ form a complete set of events accompanying event $B_1(i+1)$. Thus, the cumulative distribution function (CDF) of $B_1(i+1)$ in the energy  update process in Eq. (\ref{eq10}) can be expressed as follows

\begin{equation}\label{APPENDIX B1}
	\begin{split}
		& \mathrm{Pr}\{B_1(i+1) \leq x\} \\
		&  = \mathrm{Pr}\{B_1(i)+X(i) \leq x, TCs =\{S\}\} \\
		&  + \mathrm{Pr}\{B_1(i)+X(i) \leq x, TCs = \{S,R2\}\}\\
		&  + \mathrm{Pr}\{B_1(i)+X(i) \leq x, \big(TCs = \{S,R1\} \big) , \big(B_1(i) < M_{R1}\big)\}\\
		&  + \mathrm{Pr}\{B_1(i)+X(i) \leq x, \big(TCs = \{S,R1\} \big) , \big(B_1(i) \ge M_{R1},\\
		&\qquad \mathbb{C},\mathbb{D},\gamma_{SR2}(i)\ge\Gamma_{th},\big)\}\\
		&  + \mathrm{Pr}\{B_1(i)+X(i) \leq x, \big(TCs = \{S,R1\} \big) , \big(B_1(i) \ge M_{R1},\\
		&\qquad \mathbb{C},\mathbb{D}, \gamma_{SR2}(i)<\Gamma_{th},\gamma_{R1R2}(i)<\Gamma_{th}\big)\}\\
		&  + \mathrm{Pr}\{B_1(i)+X(i) \leq x, \big(TCs = \{S,R1,R2\} \big) ,\\
		& \qquad \big(B_2(i) \ge M_{R2},\mathbb{O}, \overline{\mathbb{F}} \big)\}\\
		&  + \mathrm{Pr}\{B_1(i)+X(i) \leq x, \big(TCs = \{S,R1,R2\} \big) ,\\
		& \qquad \big(B_2(i) \ge M_{R2},\mathbb{O}, \mathbb{F},B_1(i) \ge M_{R1},\mathbb{G} \big)\}\\
		&  + \mathrm{Pr}\{B_1(i)+X(i) \leq x, \big(TCs = \{S,R1,R2\} \big) ,\\
		&\qquad \big(B_2(i) \ge M_{R2},\mathbb{O}, \mathbb{F},B_1(i) < M_{R1} \big)\}\\
		&  + \mathrm{Pr}\{B_1(i)-M_{R1}+X(i) \leq x,\\
		&\qquad \big(TCs = \{S,R1\} \big) , \big(B_1(i) \ge M_{R1}\big), \mathbb{C} , \overline{\mathbb{D}}\}\\
		&  + \mathrm{Pr}\{B_1(i)-M_{R1}+X(i) \leq x,\\
		&\qquad \big(TCs = \{S,R1\} \big) , \big(\big(B_1(i) \ge M_{R1}\big), \mathbb{C} , \mathbb{D},\\
		&\qquad  \gamma_{SR2}(i)<\Gamma_{th},\gamma_{R1R2}(i)\ge\Gamma_{th}\big)\\
		&  + \mathrm{Pr}\{B_1(i)-M_{R1}+X(i) \leq x, \big(TCs = \{S,R1,R2\}\big) ,\\
		&\qquad \big(B_2(i) \ge M_{R2} , \mathbb{O} , \mathbb{F} ,\big(B_1(i) \ge M_{R1}\big),\overline{\mathbb{G}}\big)\}\\
		&  + \mathrm{Pr}\{B_1(i)-M_{R1}+X(i) \leq x,\big(TCs = \{S,R1,R2\}\big), \\
		&\qquad \big(B_2(i) < M_{R2} , \mathbb{O} , \big(B_1(i) \ge M_{R1}\big),\overline{\mathbb{G}}\big)\}.
	\end{split}
\end{equation}
Let $\mathrm{Pr}\{B_1(i+1) \leq x\} = G^{i+1}_1(x)$. With $i \to \infty$, if the energy buffer is in its steady state, it obeys that, $\mathrm{Pr}\{B_1(i+1) \leq x\} = G^{i+1}_1(x) = G^{i}_1(x) = G_1(x)$. In this scenario, Eq. (\ref{APPENDIX B1}) can be written as follows
\begin{equation}\label{APPENDIX B3}
	\begin{split}
		G_1(x) & = \Big[ p_S+p_{SR2}+((p_1+p_2+p_3)*o(1-f))PU2\\
		&+(p_1+p_2+p_3)(1-h)\Big] \int_{\mu_1=0}^{x} F_X(x-\mu_1)g_1(\mu_1)\, d\mu_1 \\
		&+ \Big[p_{SR1}[\mathfrak{c}*\mathfrak{d}(\mathrm{Pr}\{\gamma_{SR2}(i)\ge \Gamma_{th}\}\\
		&+\mathrm{Pr}\{\gamma_{SR2}(i)< \Gamma_{th}\}\mathrm{Pr}\{\gamma_{R1R2}(i)< \Gamma_{th}\})+1-\mathfrak{c}]\\
		&+ ((p_1+p_2+p_3)*\mathfrak{o}*\mathfrak{f}*\mathfrak{g})PU2\\
		&+((p_1+p_2+p_3)*\mathfrak{o}*\mathfrak{g})(1-PU2)\Big]\\
		& \times \int_{\mu_1=M_{R1}}^{x} F_X(x-\mu_1)g_1(\mu_1)d\mu_1\\
		& +\Big[p_{SR1}+PU2((p_1+p_2+p_3)*\mathfrak{o}*\mathfrak{f})\\
		&+(1-PU2)((p_1+p_2+p_3)*\mathfrak{o})\Big]\\
		& \times \int_{\mu_1=0}^{min(x,M_{R1})} F_X(x-\mu_1)g_1(\mu_1)\, d\mu_1 \\
		& +\Big[p_{SR1}*\mathfrak{c}((1-\mathfrak{d})+\mathfrak{d}(1-e^{-W_{SR2}\Gamma_{th}})e^{-W_{R1R2}\Gamma_{th}})\\
		&+PU2((p_1+p_2+p_3)*\mathfrak{o}*\mathfrak{f}(1-\mathfrak{g}))\\
		&+(1-PU1)((p_1+p_2+p_3)*\mathfrak{o}(1-\mathfrak{g})\Big] \\
		& \times \int_{\mu_1=M_{R1}}^{x+M_{R1}} F_X(x+M_{R1}-\mu_1)d\mu_1,
	\end{split}
\end{equation}
where $g_1(x)$ is the PDF of $B_1(i)$.$F_X(x)=1-e^{-\lambda_1 x}$ is the CDF of $X$. By simplifying Eq. (\ref{APPENDIX B3}), we arrive at
\begin{equation}\label{APPENDIX B4}
	G_1(x) = \begin{cases}
		G_{11}(x) , & 0 \leq x < M_{R1} \\
		G_{12}(x), & x \ge M_{R1}
	\end{cases}
\end{equation}
where,
\begin{equation}\label{APPENDIX B5}
	\begin{split}
		G_{11}(x) & = \int_{\mu_1=0}^{x} F_X(x-\mu_1)g_1(\mu_1)\, d\mu_1 \\
		& \quad + b_1 \int_{\mu_1=M_{R1}}^{x+M_{R1}} F_X(x+M_{R1}-\mu_1)g_1(\mu_1)\, d\mu_1, \\
		& \qquad \qquad \qquad \qquad \qquad \qquad \qquad 0\leq x < M_{R1}
	\end{split}
\end{equation}
\begin{equation}\label{APPENDIX B6}
	\begin{split}
		G_{12}(x) & = \int_{\mu_1=0}^{M_{R1}} F_X(x-\mu_1)g_1(\mu_1)\, d\mu_1 \\
		& \quad + a_1 \int_{\mu_1=M_{R1}}^{x} F_X(x-\mu_1)g_1(\mu_1)\, d\mu_1 \\
		& \quad + b_1\int_{\mu_1=M_{R1}}^{x+M_{R1}} F_X(x+M_{R1}-\mu_1)g_1(\mu_1)\, d\mu_1,\\
		& \qquad \qquad \qquad \qquad \qquad \qquad \qquad \quad x \ge M_{R1}
	\end{split}
\end{equation}
where,
\begin{equation}\label{APPENDIX B9}
	\begin{split}
		a_1&=p_S+p_{SR2}+((p_1+p_2+p_3)*\mathfrak{o}(1-\mathfrak{f}))*PU2\\
		&+(p_1+p_2+p_3)(1-e)+p_{SR1}(\mathfrak{c}*\mathfrak{d}(e^{-W_{SR2}\Gamma_{th}}\\
		&+(1-e^{-W_{SR2}\Gamma_{th}})(1-e^{-W_{R1R2}\Gamma_{th}}))\\
		&+(1-\mathfrak{c}))+((p_1+p_2+p_3)*\mathfrak{o}*\mathfrak{f}*\mathfrak{g})PU2\\
		&+((p_1+p_2+p_3)*\mathfrak{o}*\mathfrak{g})(1-PU2),
	\end{split}
\end{equation}
\begin{equation}\label{APPENDIX B10}
	\begin{split}
		b_1  & =p_{SR1}*\mathfrak{c}((1-\mathfrak{d})+\mathfrak{d}(1-e^{-W_{SR2}\Gamma_{th}})e^{-W_{R1R2}\Gamma_{th}}\\
		&+PU2((p_1+p_2+p_3)*\mathfrak{o}*\mathfrak{f}(1-\mathfrak{g})\\
		&+(1-PU2)*((p_1+p_2+p_3)*\mathfrak{o}(1-\mathfrak{g})).
	\end{split}
\end{equation}
According to Eq. (\ref{APPENDIX B4}), the PDF $g_1(x)$ may be defined as
\begin{equation}\label{APPENDIX B7}
	g_1(x) = \begin{cases}
		g_{11}(x) , & 0 \leq x < M_{R1} \\
		g_{12}(x), & x \ge M_{R1}.
	\end{cases}
\end{equation}
After substituting Eq. (\ref{APPENDIX B7}) into Eq. (\ref{APPENDIX B6}), the derivatives of Eq. (\ref{APPENDIX B6}) about $x$ can be obtained
\begin{equation}\label{APPENDIX B8}
	\begin{split}
		g_{12}(x) & = \int_{\mu_1=0}^{M_{R1}} f_X(x-\mu_1)g_{11}(\mu_1)\, d\mu_1 \\
		& \quad + a_1 \int_{\mu_1=M_{R1}}^{x} f_X(x-\mu_1)g_{12}(\mu_1)\, d\mu_1\\ 
		& \quad+ b_1  \int_{\mu_1=M_{R1}}^{x+M_{R1}} f_X(x+M_{R1}-\mu_1)g_{12}(\mu_1)\, d\mu_1,\\
		& \qquad \qquad \qquad \qquad \qquad \qquad \qquad \quad x \ge M_{R1},
	\end{split}
\end{equation}
Like $g_{22}(x)$ in appendix B, let $g_{12}(x) = k_1 e^{Q_1 x}$. Substituting $g_{12}(x) = k_1 e^{Q_1 x}$ and $f_X(x) = \lambda_1 e^{-\lambda_1 x}$ into Eq. (\ref{APPENDIX B8}), we get
\begin{equation}\label{APPENDIX B11}
	\begin{split}
		k_1 e^{Q_1 x} & = \int_{\mu_1=0}^{M_{R1}} \lambda_1 e^{-\lambda_1 (x-\mu_1)} g_{11}(\mu_1)\, d\mu_1 \\
		& \quad + \left(1-b_1 \right) \int_{\mu_1=M_{R1}}^{x} \lambda_1 e^{-\lambda_1 (x-\mu_1)} k_1 e^{Q_1 \mu_1}\, d\mu_1 + b_1 \\
		& \quad \times \int_{\mu_1=M_{R1}}^{x+M_{R1}} \lambda_1 e^{-\lambda_1 (x+M_{R1}-\mu_1)} k_1 e^{Q_1 \mu_1}\, d\mu_1, x \ge M_{R1}
	\end{split}
\end{equation}
Simplifying Eq. (\ref{APPENDIX B11}), we have
\begin{equation}\label{APPENDIX B12}
	\begin{split}
		k_1 e^{Q_1 x} & = \lambda_1 e^{-\lambda_1 x}\int_{\mu_1=0}^{M_{R1}} e^{\lambda_1 \mu_1} g_{11}(\mu_1)\, d\mu_1 \\
		& \quad - \dfrac{\lambda_1 k_1 \left[ b_1 + \left(1-b_1 \right) e^{\lambda_1 M_{R1}}\right]}{\lambda_1 + Q_1} e^{\left(Q_1 M_{R1} - \lambda_1 x\right)} \\
		& \quad + \dfrac{b_1 \lambda_1 e^{Q_1 M_{R1}} + \left(1-b_1 \right) \lambda_1}{\lambda_1 + Q_1} k_1 e^{Q_1 x}, x \ge M_{R1}
	\end{split}
\end{equation}
In order to ensure $g_{12}(x) = k_1 e^{Q_1 x}$, the coefficients of corresponding terms in Eq. (\ref{APPENDIX B12}) both sides are equal. Thus, we get
\begin{subequations}
	\begin{numcases}{}
		\dfrac{b_1 \lambda_1 e^{Q_1 M_{R1}} + a_1 \lambda_1}{\lambda_1 + Q_1} = 1, \label{APPENDIX B13a}\\ %\label{}中 括号内容只是个标志，引用时对应就行
		\begin{aligned}&\dfrac{k_1 e^{Q_1 M_{R1}} \left[ b_1 + a_1 e^{\lambda_1 M_{R1}}\right]}{\lambda_1 + Q_1}\\
			&\qquad \qquad = \int_{\mu_1=0}^{M_{R1}} e^{\lambda_1 \mu_1} g_{11}(\mu_1)\, d\mu_1. \label{APPENDIX B13b}
		\end{aligned}
	\end{numcases}
\end{subequations}
Although, it can be easily found from Eq. (\ref{APPENDIX B13a}) that $Q_{1_0} = 0$ is one of the solutions of $Q_1$ in Eq. (\ref{APPENDIX B13a}), $Q_{1_0}$ does not satisfy the condition that $g_{12}(x)$ is a limiting distribution. And the other solution $Q_{1_1}$ of $Q_1$ in Eq. (\ref{APPENDIX B13a}) can be obtained by simplifying Eq. (\ref{APPENDIX B13a}) as
\begin{equation}\label{APPENDIX B14}
	b_1 \lambda_1 e^{Q_1 M_{R1}} = \lambda_1 - a_1 \lambda_1 + Q_1= b_1 \lambda_1 + Q_1,
\end{equation}
With Lambert W function, the solution $Q_{1_1}$ can be expressed as follows
\begin{equation}\label{APPENDIX B15}
	Q_{1} = \frac{-W\left(-b_1\lambda_1 M_{R1} e^{-b_1\lambda_1 M_{R1}}\right)}{M_{R1}}-b_1\lambda_1,
\end{equation}
Due to the property of Lambert W function, when $b_1\lambda_1 M_{R1} \leq 1$, $W\left(-b_1\lambda_1 M_{R1} e^{-b_1\lambda_1 M_{R1}}\right)= -b_1\lambda_1 M_{R1}$, as a result, $Q_{1_1}=Q_{1_0}=0$. On the other hand, when $b_1\lambda_1 M_{R1} >1$, $W\left(-b_1\lambda_1 M_{R1} e^{-b_1\lambda_1 M_{R1}}\right)> -b_1\lambda_1 M_{R1}$ so that $Q_{1_1}<0$, ensuring the limiting distribution of $g_{12}(x)$. Thus, for the stationary distribution of $g_{12}(x)$, we obtain
\begin{equation}\label{APPENDIX B16}
	Q_1 = \frac{-W\left(-b_1\lambda_1 M_{R1} e^{-b_1\lambda_1 M_{R1}}\right)}{M_{R1}}-b_1\lambda_1,\quad b_1\lambda_1 M_{R1}>1.
\end{equation}

Similarly, when $0\leq x <M_{R1}$, substituting Eq. (\ref{APPENDIX B7}) into Eq. (\ref{APPENDIX B5}), the derivatives of Eq. (\ref{APPENDIX B5}) about $x$ can be obtained
\begin{equation}\label{APPENDIX B17}
	\begin{split}
		g_{11}(x) & = b_1 \int_{\mu_1=M_{R1}}^{x+M_{R1}} f_X(x+M_{R1}-\mu_1)g_{12}(\mu_1)\, d\mu_1\\
		& \quad + \int_{\mu_1=0}^{x} f_X(x-\mu_1)g_{11}(\mu_1)\, d\mu_1, 0 \leq x <M_{R1}
	\end{split}
\end{equation}
Substituting $g_{12}(x)=k_1 e^{Q_1 x}$ and $f_X(x)=\lambda_1 e^{-\lambda_1 x}$ into Eq. (\ref{APPENDIX B17}), we get
\begin{equation}\label{APPENDIX B18}
	\begin{split}
		g_{11}(x) & = \lambda_1 \int_{\mu_1=0}^{x} e^{-\lambda_1 \left(x-\mu \right)}g_{11}(\mu_1)\, d\mu_1 \\
		& \quad + \dfrac{b_1 k_1 \lambda_1 e^{Q_1 M_{R1}}}{\lambda_1 + Q_1} \left(e^{Q_1 x}-e^{-\lambda_1 x} \right), 0 \leq x <M_{R1}
	\end{split}
\end{equation}
Similar to $ g_{21}(x)$ in appendix B, the solution of $g_{11}(x)$ may be given as follows
\begin{equation}\label{APPENDIX B21}
	\begin{split}
		g_{11}(x) & = \dfrac{b_1 k_1 \lambda_1 e^{Q_1 M_{R1}}}{\lambda_1 + Q_1} \left(e^{Q_1 x}-e^{-\lambda_1 x} \right) \\
		& \quad + \lambda_1 \int_{t=0}^{x} \frac{b_1 k_1 \lambda_1 e^{Q_1 M_{R1}}}{\lambda_1 + Q_1} \left(e^{Q_1 t}-e^{-\lambda_1 t} \right)\, dt \\
		& = \dfrac{b_1 k_1 \lambda_1 e^{Q_1 M_{R1}}\left(e^{Q_1 x}-1 \right)}{Q_1}, 0 \leq x <M_{R1}
	\end{split}
\end{equation}
According to the unit area condition on $g_1(x)$, we have
\begin{equation}\label{APPENDIX B22}
	\int_{x=0}^{\infty} g_1(x)\, dx = \int_{x=0}^{M_{R1}} g_{11}(x)\, dx +\int_{x=M_{R1}}^{\infty} g_{12}(x)\, dx = 1,
\end{equation}
Substituting $g_{11}(x) = \frac{b_1 k_1 \lambda_1 e^{Q_1 M_{R1}}\left(e^{Q_1 x}-1 \right)}{Q_1}$ and $g_{12}(x) = k_1 e^{Q_1 x}$ into Eq. (\ref{APPENDIX B22}), we get
\begin{equation}\label{APPENDIX B23}
	\dfrac{b_1 k_1 \lambda_1 e^{Q_1 M_{R1}}}{Q_1} \int_{x=0}^{M_{R1}} \left(e^{Q_1 x}-1 \right)\, dx +  k_1 \int_{x=M_{R1}}^{\infty} e^{Q_1 x}\, dx =1,
\end{equation}
Simplifying Eq. (\ref{APPENDIX B23}), we have
\begin{equation}\label{APPENDIX B24}
	\dfrac{b_1 k_1 \lambda_1 e^{Q_1 M_{R1}}}{Q_1} \left[\dfrac{e^{Q_1 M_{R1}}-1}{Q_1}-M_{R1} \right] - \dfrac{k_1 e^{Q_1 M_{R1}}}{Q_1}=1,
\end{equation}
Substituting Eq. (\ref{APPENDIX B14}) into Eq. (\ref{APPENDIX B24}), then simplifying Eq. (\ref{APPENDIX B24}), the value of $k_1$ can be obtained as follows
\begin{equation}\label{APPENDIX B25}
	k_1 = \dfrac{-Q_1}{M_{R1} \left(b_1 \lambda_1 + Q_1\right)},
\end{equation}
Substituting Eq. (\ref{APPENDIX B14}) and Eq. (\ref{APPENDIX B25}) into Eq. (\ref{APPENDIX B21}), we arrive at
\begin{equation}\label{APPENDIX B26}
	g_{11}(x) =  \dfrac{1-e^{Q_1 x}}{M_{R1}}.
\end{equation}
Substituting Eq. (\ref{APPENDIX B26}) into the right side of Eq. (\ref{APPENDIX B13b}), we obtain
\begin{equation}\label{APPENDIX B27}
	\begin{split}
		\int_{\mu_1=0}^{M_{R1}} e^{\lambda_1 \mu_1} g_{11}(\mu_1)\, d\mu_1 & = \int_{\mu_1=0}^{M_{R1}} \dfrac{\left(1-e^{Q_1 x}\right)e^{\lambda_1 \mu_1}}{M_{R1}}\, d\mu_1 \\
		& = \dfrac{1-e^{\left(\lambda_1+Q_1\right)M_{R1}}}{\left(\lambda_1+Q_1\right)M_{R1}} - \dfrac{1-e^{\lambda_1 M_{R1}}}{\lambda_1 M_{R1}},
	\end{split}
\end{equation}
The equation in Eq. (\ref{APPENDIX B14}) leads us to conclude $\lambda_1 M_{R1} = \frac{\left(\lambda_1+Q_1\right)M_{R1}}{b_1 e^{Q_1 M_{R1}} + a_1}$. Substituting this conclusion in Eq. (\ref{APPENDIX B27}), we have
\begin{equation}\label{APPENDIX B28}
	\begin{split}
		\int_{\mu_1=0}^{M_{R1}} e^{\lambda_1 \mu_1} g_{11}(\mu_1)\, d\mu_1 = \dfrac{\left(1-e^{Q_1 M_{R1}}\right)\left(b_1+a_1 e^{\lambda_1 M_{R1}}\right)}{\left(\lambda_1+Q_1\right)M_{R1}},
	\end{split}
\end{equation}
Similarly, the conclude $1-e^{Q_1 M_{R1}} = \frac{-Q_1}{b_1\lambda_1}$ may be obtained from Eq. (\ref{APPENDIX B14}). Substituting this conclusion in Eq. (\ref{APPENDIX B28}), we arrive at
\begin{equation}\label{APPENDIX B29}
	\begin{split}
		\int_{\mu_1=0}^{M_{R1}} e^{\lambda_1 \mu_1} g_{11}(\mu_1)\, d\mu_1 & = \dfrac{-Q_1 \left(b_1+a_1 e^{\lambda_1 M_{R1}}\right)}{\left(\lambda_1+Q_1\right)M_{R1} b_1\lambda_1} \\
		& = \dfrac{-Q_1 e^{Q_1 M_{R1}} \left(b_1+a_1 e^{\lambda_1 M_{R1}}\right)}{\left(\lambda_1+Q_1\right)M_{R1} b_1\lambda_1 e^{Q_1 M_{R1}}} \\
		& = \dfrac{-Q_1 e^{Q_1 M_{R1}} \left(b_1+a_1 e^{\lambda_1 M_{R1}}\right)}{M_{R1} \left(b_1\lambda_1 + Q_1\right) \left(\lambda_1+Q_1\right)} \\
		& = \dfrac{k_1 e^{Q_1 M_{R1}} \left(b_1+a_1 e^{\lambda_1 M_{R1}}\right)}{\lambda_1+Q_1}.
	\end{split}
\end{equation}
Now, it can be shown that $g_{11}(x)$ satisfies the condition in Eq. (\ref{APPENDIX B13b}). Therefore, there is no doubt that the unique solution $g_{11}(x)$ in Eq. (\ref{APPENDIX B26}) for Eq. (\ref{APPENDIX B17}) and the unique solution $g_{12}(x) = k_1 e^{Q_1 x}$ for Eq. (\ref{APPENDIX B8}) are obtained.

%\section*{Acknowledgment}

%The authors would like to thank...

\ifCLASSOPTIONcaptionsoff
  \newpage
\fi

%\begin{IEEEbiography}{Yuguang ``Michael'' Fang}
%Biography text here.
%\end{IEEEbiography}

%It is not necessary to upload the biography when you submit your manuscript.


\begin{thebibliography}{1}

\bibitem{an2022opportunistic}
Wannian An, Chen Dong, Xiaodong Xu, Chao Xu, Shujun Han, and Lei Teng.
\newblock Opportunistic routing aided cooperative communication network with
  energy-harvesting nodes, 2022.

\end{thebibliography}


\begin{thebibliography}{19}
% You can use other form of bib file by changing here...
\bibitem{article1}M. -L. Ku, W. Li, Y. Chen and K. J. Ray Liu, "Advances in Energy Harvesting Communications: Past, Present, and Future Challenges," \textit{IEEE Communications Surveys \& Tutorials}, vol. 18, no. 2, pp. 1384-1412, Secondquarter 2016.
\bibitem{article2}J. Ramis-Bibiloni and L. Carrasco-Martorell, "Energy Harvesting Effect on the Sensors Battery Lifespan of an Energy Efficient SmartBAN Network," in \textit{2021 International Wireless Communications and Mobile Computing (IWCMC)}, 2021, pp. 1593-1598.
\bibitem{3}D. Sui, F. Hu, W. Zhou, M. Shao and M. Chen, "Relay Selection for Radio Frequency Energy-Harvesting Wireless Body Area Network With Buffer," \textit{IEEE Internet of Things Journal}, vol. 5, no. 2, pp. 1100-1107, April 2018.
\bibitem{4}A. Alsharoa, H. Ghazzai, A. E. Kamal and A. Kadri, "Optimization of a Power Splitting Protocol for Two-Way Multiple Energy Harvesting Relay System," \textit{IEEE Transactions on Green Communications and Networking}, vol. 1, no. 4, pp. 444-457, Dec. 2017.
\bibitem{5}Y. Zou, J. Zhu and X. Jiang, "Joint Power Splitting and Relay Selection in Energy-Harvesting Communications for IoT Networks," \textit{IEEE Internet of Things Journal}, vol. 7, no. 1, pp. 584-597, Jan. 2020.
\bibitem{article3}R. Morsi, D. S. Michalopoulos and R. Schober, "On-off transmission policy for wireless powered communication with energy storage," in \textit{2014 48th Asilomar Conference on Signals, Systems and Computers}, 2014, pp. 1676-1682.
\bibitem{article4}R. Morsi, D. S. Michalopoulos and R. Schober, "Performance Analysis of Near-Optimal Energy Buffer Aided Wireless Powered Communication," \textit{IEEE Transactions on Wireless Communications}, vol. 17, no. 2, pp. 863-881, Feb. 2018.
\bibitem{article5}Y. Wu, L. p. Qian, L. Huang and X. Shen, "Optimal Relay Selection and Power Control for Energy-Harvesting Wireless Relay Networks," \textit{IEEE Transactions on Green Communications and Networking}, vol. 2, no. 2, pp. 471-481, June 2018.
\bibitem{article6}D. Bapatla and S. Prakriya, "Performance of Energy-Buffer Aided Incremental Relaying in Cooperative Networks," \textit{IEEE Transactions on Wireless Communications}, vol. 18, no. 7, pp. 3583-3598, July 2019.
\bibitem{article7}D. Bapatla and S. Prakriya, "Performance of a Cooperative Communication Network With Green Self-Sustaining Nodes," \textit{IEEE Transactions on Green Communications and Networking}, vol. 5, no. 1, pp. 426-441, March 2021.
\bibitem{article8}V. P. Tuan, S. Q. Nguyen and H. Y. Kong, "Performance analysis of energy-harvesting relay selection systems with multiple antennas in presence of transmit hardware impairments," in \textit{2016 International Conference on Advanced Technologies for Communications (ATC)}, 2016, pp. 126-130.
\bibitem{article9}B. S. Awoyemi, A. S. Alfa and B. T. Maharaj, "Network Restoration in Wireless Sensor Networks for Next-Generation Applications," \textit{IEEE Sensors Journal}, vol. 19, no. 18, pp. 8352-8363, 15 Sept.15, 2019.
\bibitem{article10}W. Lu et al., "OFDM based bidirectional multi-relay SWIPT strategy for 6G IoT networks," \textit{China Communications}, vol. 17, no. 12, pp. 80-91, Dec. 2020.
\bibitem{article11}N. Chakchouk, "A Survey on Opportunistic Routing in Wireless Communication Networks," \textit{IEEE Communications Surveys \& Tutorials}, vol. 17, no. 4, pp. 2214-2241, Fourthquarter 2015.

\bibitem{article12}J. Zuo, C. Dong, S. X. Ng, L. Yang and L. Hanzo, "Cross-Layer Aided Energy-Efficient Routing Design for Ad Hoc Networks," \textit{IEEE Communications Surveys \& Tutorials}, vol. 17, no. 3, pp. 1214-1238, thirdquarter 2015.
\bibitem{article13}J. Zuo, C. Dong, H. V. Nguyen, S. X. Ng, L. Yang and L. Hanzo, "Cross-Layer Aided Energy-Efficient Opportunistic Routing in Ad Hoc Networks," \textit{IEEE Transactions on Communications}, vol. 62, no. 2, pp. 522-535, February 2014.
\bibitem{anwannian}Wannian An and Chen Dong and Xiaodong Xu and Chao Xu and Shujun Han and Lei Teng, "Opportunistic Routing aided Cooperative Communication Network with Energy-Harvesting Nodes", \emph{arXiv e-print}, May. 2022, arXiv:2205.06482.
\bibitem{article14}Zhuo Chen, Jinhong Yuan and B. Vucetic, "Analysis of transmit antenna selection/maximal-ratio combining in Rayleigh fading channels," \textit{IEEE Transactions on Vehicular Technology}, vol. 54, no. 4, pp. 1312-1321, July 2005.
\bibitem{article15}
S. Luo, R. Zhang and T. J. Lim, "Optimal Save-Then-Transmit Protocol for Energy Harvesting Wireless Transmitters," \textit{IEEE Transactions on Wireless Communications}, vol. 12, no. 3, pp. 1196-1207, March 2013.
\bibitem{article16}
B. Zhang, C. Dong, M. El-Hajjar and L. Hanzo, "Outage Analysis and Optimization in Single- and Multiuser Wireless Energy Harvesting Networks," \textit{IEEE Transactions on Vehicular Technology}, vol. 65, no. 3, pp. 1464-1476, March 2016.
\bibitem{30}
A. Polyanin and A. Manzhirov, \textit{Handbook of Integral Equations: Second Edition, ser. Handbooks of mathematical equations}. Taylor \& Francis, 2008.
\bibitem{29}
Loynes, R. "The stability of a queue with non-independent inter-arrival and service times". \textit{Mathematical Proceedings of the Cambridge Philosophical Society}, vol. 58, no. 3, pp. 497-520, Jul. 1962.
\end{thebibliography}
\end{document}